\documentclass[journal]{IEEEtran}
\usepackage{amsmath,amssymb,amsthm}
\usepackage{graphicx}
\usepackage{verbatim}
\usepackage{multirow,cite}
\usepackage{epsfig}
\usepackage{algorithm,algorithmic}
\usepackage{balance}

\newtheorem{thm}{Theorem}[section]
\newtheorem{lemma}{Lemma}[section]

\def\d{\downarrow}
\def\u{\uparrow}

\begin{document}
\title{New Crosstalk Avoidance Codes\\ Based on a Novel Pattern Classification}

\author{Feng~Shi,~\IEEEmembership{Student Member,~IEEE,} Xuebin~Wu, and Zhiyuan~Yan~\IEEEmembership{Senior Member,~IEEE}

\thanks{F. Shi and Z. Yan are with the Department of Electrical and Computer Engineering, Lehigh University, Bethlehem, PA 18015 (e-mails: \{fes209, yan\}@lehigh.edu). X. Wu is with LSI Corporation in Milpitas, CA, USA (e-mail: xuebin.wu@lsi.com).}}

\maketitle
\begin{abstract}
The crosstalk delay associated with global on-chip interconnects becomes more severe in deep submicron technology, and hence can greatly affect the overall system performance. Based on a delay model proposed by Sotiriadis \emph{et al.}, transition patterns over a bus can be classified according to their delays. Using this classification, crosstalk avoidance codes (CACs) have been proposed to alleviate the crosstalk delays by restricting the transition patterns on a bus.
In this paper, we first propose a new classification of transition patterns, and then devise a new family of CACs based on this classification. In comparison to the previous classification, our classification has more classes and the delays of its classes do not overlap, both leading to more accurate control of delays. Our new family of CACs includes some previously proposed codes as well as new codes with reduced delays and improved throughput. Thus, this new family of crosstalk avoidance codes provides a wider variety of tradeoffs between bus delay and efficiency.
Finally, since our analytical approach to the classification and CACs treats the technology-dependent parameters as variables, our approach can be easily adapted to a wide variety of technology.
\end{abstract}
\section{INTRODUCTION}
\label{sec:intro}
\IEEEPARstart{R}{ecent} International Technology Roadmap of Semiconductors (ITRS)~\cite{ITRS} has shown a troubling trend: while gate delay \textbf{decreases} with scaling, global wire delay \textbf{increases}. This is because with the process technologies scaling down into deep submicrometer (DSM), the crosstalk delay becomes dominant in global wire delay due to the increasing coupling capacitance between adjacent wires. Hence, the crosstalk delay has become a serious bottleneck of the overall system performance.

The analytical model proposed by Sotiriadis \emph{et al.}~\cite{Sot01,Sot02}, a widely used delay model, gives upper bounds on the delay of all wires on a bus. According to \cite{Sot01,Sot02},
the delay of the $k$-th wire ($k \in \left\{1,2,\cdots,m \right\}$) of an $m$-bit bus is given by
\begin{equation}
  T_k = \left\{
    \begin{array}{ll}
      \tau_0[(1+\lambda)\Delta_1^2-\lambda\Delta_1\Delta_2], & k=1\\
      \tau_0[(1+2\lambda)\Delta_k^2-\lambda\Delta_k(\Delta_{k-1}+\Delta_{k+1})],
      & k\neq 1,m\\
      \tau_0[(1+\lambda)\Delta_m^2-\lambda\Delta_m\Delta_{m-1}],& k=m,\\
    \end{array}
  \right.
  \label{Eq:1}
\end{equation}
where $\lambda$ is the ratio of the coupling capacitance between
adjacent wires and the ground capacitance,
$\tau_0$ is the propagation delay of a wire free of crosstalk,
and $\Delta_k$ is 1 for 0 $\rightarrow$ 1 transition, -1 for
1 $\rightarrow$ 0 transition, or 0 for no transition on the $k$-th
wire. In this model, the delay of the $k$-th wire
depends on the transition patterns of at most three wires, $k-1$, $k$, and $k+1$ only.
The transition patterns over these three wires can be classified based on Eq.~(\ref{Eq:1}) into five classes, denoted by $Di$ for $i=0,1,2,3,4$, and the patterns in $Di$ have a worst-case delay  $(1+i\lambda)\tau_0$. This classification enables one to limit the worst-case delay over a bus by restricting the patterns transmitted on the bus. That is, by avoiding all transition patterns in $Di$ for $i > i_0$, one can achieve a worst-case delay of $(1+i_0\lambda)\tau_0$ over the bus. Based on this principle, crosstalk avoidance codes (CACs) of different worst-case delays have been proposed (see, for example, \cite{Dua01,Dua04,Vic01}). For example, forbidden overlap codes (FOCs), forbidden transition codes (FTCs),  forbidden pattern codes (FPCs), and one lambda codes (OLCs) achieve a worst-case delay of $(1+3\lambda)\tau_0$, $(1+2\lambda)\tau_0$, $(1+2\lambda)\tau_0$, and $(1+\lambda)\tau_0$, respectively. Based on Eq.~(\ref{Eq:1}), a worst-case delay of
$\tau_0$ can be achieved  by assigning two protection wires to each data wire~\cite{Dua04}.
Other types of CACs, such as those with equalization \cite{Sri08} or two-dimensional CACs \cite{Wu08}, have been proposed in the literature.
For CACs, since the area and power consumption of their encoder/decoder (CODECs) are all overheads, the complexities of the CODECs are important to the effectiveness of CACs. Thus, efficient CODECs have been proposed for CACs \cite{Dua08,Dua09,WY_TVLSI11}.

The classification of transition patterns based on the model in~\cite{Sot01,Sot02} has two drawbacks.
First, the model in~\cite{Sot01,Sot02} has limited accuracy because of its dependence on only three wires: the model overestimates the delays of patterns in $D1$ through $D4$, while it underestimates the delays of patterns in $D0$. For this reason, the scheme with a worst-case delay of
$\tau_0$ in \cite{Dua04} is invalid since its actual delay is much greater.
Second, the actual delay ranges in some classes overlap with others.
This, plus the overestimation of delays for $D1$ through $D4$, implies that the delays of existing CACs are not tightly controlled.
These drawbacks motivate us to include more wires and to classify the transition patterns without overlapping delay ranges.

In \cite{SWY_sips11_model}, we have proposed a new analytical five-wire delay model.  Two extra neighboring wires are included in the delay model \cite{SWY_sips11_model}, and the delay of the middle wire of five neighboring wires is determined by the transition patterns on all five wires. This five-wire model has better accuracy than the model in~\cite{Sot01,Sot02} for $Di$ for $i=0,1,2,3,4$ \cite{SWY_sips11_model}. This work confirms that using more wires leads to improved accuracy.

There are two main contributions in this paper: \begin{itemize}
\item  First, we approximate the crosstalk delay in a five-wire model and propose a new classification of transition patterns.
\item  Second, we propose a family of CACs based on our classification. \end{itemize}

The work in this paper is different from previous works, including our previous works, in several aspects:
\begin{itemize}
\item First, although the delay approximation in this paper is also based on a five-wire model, it is different from that in our previous work \cite{SWY_sips11_model}. The delay approximation in this paper is carried out by extending the approach in \cite{Sak93} from a three-wire model to a five-wire one.
\item Second, our classification of transition patters is different from that in~\cite{Sot01,Sot02} (based on Eq.~(\ref{Eq:1})), in two aspects. First, our classification has seven classes as opposed to five based on Eq.~(\ref{Eq:1}). Second, while the delays of some classes overlap for the classification based on Eq.~(\ref{Eq:1}), all classes in our classification have non-overlapping delays. These two key differences allow us to have a more accurate control of delays for transition patterns.
\item Our new family of CACs is also different from previously proposed CACs, all of which are based on the classification in~\cite{Sot01,Sot02} (based on Eq.~(\ref{Eq:1})). While some codes in this new family are shown to be the same as existing CACs, OLCs, FPCs, and FOCs,  this family also includes new codes that achieve smaller worst-case delays and improved throughputs than OLCs, which have the smallest worst-case delays among all existing CACs.
\end{itemize}

The rest of the paper is organized as follows. In Section~\ref{sec:classification}, we first propose our classification and compare it with that in~\cite{Sot01,Sot02}. We then present our new family of CACs in Section~\ref{sec:cacs} and compare their performance with existing CACs in Section~\ref{sec:performance}. Some concluding remarks are provided in Section~\ref{sec:conclusion}.

\section{INTERCONNECT DELAYS AND CLASSIFICATION}\label{sec:classification}
\subsection{Interconnect Modeling}
Since the functionality and performance in DSM technology are greatly affected by the parasitics, distributed RC models are widely employed to analyze on-chip interconnects. In this paper, we consider the distributed RC model of five wires shown in Fig.~\ref{fig:5dist}, where $V_i(x,t)$ denotes the transient signal at time $t$ and position $x$ ($0 \le x \le L$) over wire $i$ for $i \in \{1,2,3,4,5\}$, $r$ and $c$ denote the resistance and ground capacitance per unit length, respectively. Also, $\lambda c$ denotes the coupling capacitance per unit length between two adjacent wires. The value of $\lambda$ depends on many factors, such as the metal layer in which we route the bus, the wire width, the spacing between adjacent wires, and the distance to the ground layer. We consider a uniformly distributed bus with the same parameters $r$, $c$, and $\lambda$ for all the wires.

\begin{figure}[!tb]
\begin{minipage}[b]{1.0\linewidth}
  \centering
 \centerline{\epsfig{figure=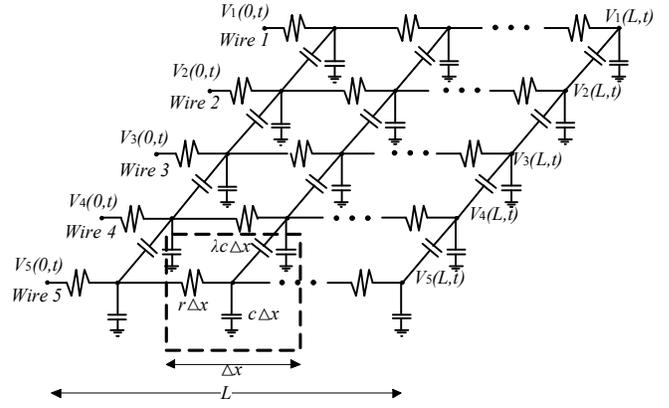,width=8.5cm}}
\end{minipage}
\caption{A distributed RC model for five wires.}
\label{fig:5dist}\vspace{-0.1in}
\end{figure}

\subsection{Derivation of Closed-form Expressions}\label{sec:dev}
When determining the delay of a wire, the model in~\cite{Sot01,Sot02} considers only the effects of either \textbf{one} or \textbf{two} neighboring wires (cf.~Eq.~(\ref{Eq:1})).
To address the drawbacks of the model in~\cite{Sot01,Sot02} described above, additional neighboring wires need to be accounted for.
In our delay derivation below,  whenever possible we consider \textbf{four} neighboring wires of a wire, two neighboring wires on each side, to determine its delay.
To approximate the delay of a side wire (wires $1$, $2$, $n-1$ or $n$) of an $n$-wire bus, \textbf{three} neighboring wires are considered. This is because the side wires are affected by fewer neighboring wires. This scheme is similar to the model in~\cite{Sot01,Sot02} and appears to work well. We focus on the 50\% delay, which is defined as the time required for the unit step response to reach 50\% of its final value.

In~\cite{Sak93}, the crosstalk of two coupled lines was described by partial differential equations (PDEs), and a technique for decoupling these highly coupled PDEs was introduced by using eigenvalues and corresponding eigenvectors. In our work, we extend this approach from a three-wire model to a five-wire one.
Specifically, we first use the  technique  in~\cite{Sak93} to decouple the PDEs that describe the crosstalk of four coupled wires, then solve these independent PDEs for closed-form expressions, and finally approximate the delays of each wire.

The PDEs characterizing five wires with length $L$ are given by:
\begin{equation}
\frac{\partial^2}{\partial x^2}\mathbf{V}(x,t)
=\mathbf{RC}\frac{\partial}{\partial t}\mathbf{V}(x,t),
\label{Eq:2}
\end{equation}
where $\mathbf{R}=\mbox{diag}\{r \; r \; r \; r \; r\}$,
$\mathbf{V}(x,t)=[V_1(x,t)\;V_2(x,t)\;V_3(x,t)\;V_4(x,t)\;V_5(x,t)]^T$, and
\[\mathbf{C}=c\left[ \begin{smallmatrix} 1+\lambda & -\lambda & 0 & 0 & 0 \\ -\lambda & 1+2\lambda & -\lambda & 0 & 0\\ 0 & -\lambda & 1+2\lambda & -\lambda & 0 \\ 0 & 0 & -\lambda & 1+2\lambda & -\lambda \\ 0 & 0 & 0 & -\lambda & 1+\lambda
 \end{smallmatrix}\right].\]

The eigenvalues of $\mathbf{C}/c$ are given by $p_1=1$, $p_2=1+\frac{5+\sqrt{5}}{2}\lambda$, $p_3=1+\frac{5-\sqrt{5}}{2}\lambda$, $p_4=1+\frac{3+\sqrt{5}}{2}\lambda$, and $p_5=1+\frac{3-\sqrt{5}}{2}\lambda$. Their corresponding eigenvectors $\mathbf{e}_i$'s are given by $\mathbf{e}_1=[1\; 1\; 1\;1\;1]^T$, $\mathbf{e}_2=[\frac{\sqrt{5}-1}{4}\; -\frac{1+\sqrt{5}}{4}\; 1\; -\frac{1+\sqrt{5}}{4}\; \frac{\sqrt{5}-1}{4}]^T$, $\mathbf{e}_3=[\frac{-\sqrt{5}+1}{4}\; \frac{\sqrt{5}-1}{4}\; 1\; \frac{\sqrt{5}-1}{4}\; -\frac{\sqrt{5}+1}{4}]^T$, $\mathbf{e}_4=[-1\;\frac{\sqrt{5}+1}{2}\; 0\; -\frac{\sqrt{5}+1}{2}\; 1]^T$, and $\mathbf{e}_5=[-1\; -\frac{\sqrt{5}-1}{2}\; 0\; \frac{\sqrt{5}-1}{2}\; 1]^T$, respectively.

With a technique for decoupling partial differential equations similar to \cite{Sak93}, Eq.~(\ref{Eq:2}) is transformed into
\begin{equation}
\frac{\partial^2}{\partial x^2}U_i(x,t)=rcp_i\frac{\partial}{\partial t} U_i(x,t), \mbox{ for }i=1,2,3,4,5,\label{Eq:3}
\end{equation}
where $U_i(x,t)=\mathbf{V}^T(x,t) \mathbf{e}_i$ denotes the transformed signals. The decoupled PDEs in Eq.~(\ref{Eq:3}) are independent of each other. Each $U_i(x,t)$ describes a single wire with a modified capacitance $cp_i$. The solution to $U_i(L,t)$ is given by a series of the form $U_i(L,t) = V_{dd} + \sum_{k=0}^\infty r_k e^{-\frac{t}{s_k \tau}}$.
As shown in~\cite{Sak93}, a single-exponent approximation $V_{dd}(1+r_0 e^{-\frac{t}{s_0 \tau}})$ is enough for $t/\tau > 0.1$, where $r_0$ and $s_0$ are the coefficients of the most significant term.

For different transitions, we solve Eq.~(\ref{Eq:3}) for $U_i(x,t)$ and obtain $V_3(L,t)=\frac{1}{5}[U_1(L,t)+2U_2(L,t)+2U_3(L,t)]$, which is given by a sum of a constant and three exponent terms, $V_{dd}(1-c_0 e^{-\frac{t}{a_0\tau}} - c_1 e^{-\frac{t}{a_1\tau}} - c_2 e^{-\frac{t}{a_2\tau}})$. Then the 50\% delay of wire 3 can be evaluated by solving $V_3(L,t)=0.5V_{dd}$.

\begin{table*}[!htb]
\caption{Closed-form expressions for the output signals on wire 3 in a five-wire bus with evaluated and simulated 50\% delays ($\tau_0=1.42$~$\mathrm{ps}$,~$\tau=\frac{8}{\pi^2}\tau_0$, $\lambda=12.24$, $a_0=1$, $a_1=1+\frac{5-\sqrt{5}}{2}\lambda$, and $a_2=1+\frac{5+\sqrt{5}}{2}\lambda$ for all classes).}\label{tab:1}
\begin{center}
\begin{tabular}{|c|l|c|c|c|c|c|}
\hline
\multirow{3}{*}{Class $i$} & \multirow{3}{*}{Patterns} & \multicolumn{3}{|c|}{Closed-form expression for output signal on wire 3} & \multirow{3}{*}{Evaluated delays (ps)} & \multirow{3}{*}{Sim. delay (ps)}\\
\cline{3-5} & & \multicolumn{3}{|c|}{$V_{dd}(1-c_0 e^{-\frac{t}{a_0\tau}} - c_1 e^{-\frac{t}{a_1\tau}} - c_2 e^{-\frac{t}{a_2\tau}})$} & &\\
\cline{3-5} & & $c_0$ & $c_1$ &  $c_2$ & &\\
\hline
\hline
\multirow{3}{*}{0} & {$\u \u \u \u \u$} & $\frac{4}{\pi}$ & 0 & 0 & 1.08 & 1.18\\
\cline{2-7}
& {-$\u \u \u \u$}, $\u \u \u \u$- & $\frac{16}{5\pi}$ & $\frac{2(1+\sqrt{5})}{5\pi}$ & $\frac{2(1-\sqrt{5})}{5\pi}$ & \textbf{1.41} & \textbf{1.50}\\
\cline{2-7}
& {$\u $-$\u \u \u$}, $\u \u \u$-$\u$ & $\frac{16}{5\pi}$ & $\frac{2(1-\sqrt{5})}{5\pi}$ & $\frac{2(1+\sqrt{5})}{5\pi}$ & \textbf{1.41} & \textbf{1.50}\\
\hline
\hline
\multirow{3}{*}{1}& {-$\u \u \u $-}, $\d\u \u \u \u$, $\u \u \u \u\d$ & $\frac{12}{5\pi}$ & $\frac{4(1+\sqrt{5})}{5\pi}$ & $\frac{4(1-\sqrt{5})}{5\pi}$ & \textbf{2.35} & 2.40\\
\cline{2-7}
& {- -$\u \u \u$}, $\u \u \u$- -, -$\u \u$-$\u$, $\u$-$ \u \u$- & $\frac{12}{5\pi}$ & $\frac{4}{5\pi}$ & $\frac{4}{5\pi}$ & \textbf{2.35} & 2.40\\
\cline{2-7}
&{$\u$-$\u$-$\u$}, $\u\u\u\d\u$, $\u\d\u\u\u$ & $\frac{12}{5\pi}$ & $\frac{4(1-\sqrt{5})}{5\pi}$ & $\frac{4(1+\sqrt{5})}{5\pi}$ &\textbf{2.35} & \textbf{2.45}\\
\hline
\hline
\multirow{3}{*}{2}& {-$\u \u \u \d$}, $\d \u \u \u $- & $\frac{8}{5\pi}$ & $\frac{6(1+\sqrt{5})}{5\pi}$ & $\frac{6(1-\sqrt{5})}{5\pi}$ & 6.17 & 6.84\\
\cline{2-7}
& {- -$\u \u$-}, -$\u \u$- -, $\d$-$\u \u\u$, $\d\u \u$-$\u$, & \multirow{2}{*}{$\frac{8}{5\pi}$} & \multirow{2}{*}{$\frac{2(3+\sqrt{5})}{5\pi}$} & \multirow{2}{*}{$\frac{2(3-\sqrt{5})}{5\pi}$} & \multirow{2}{*}{9.62} & \multirow{2}{*}{9.21}\\
&$\u$-$\u \u\d$, $\u\u \u$-$\d$ & & & & & \\
\cline{2-7}
& {$\d \u \u \u \d$} & $\frac{4}{5\pi}$ & $\frac{8(1+\sqrt{5})}{5\pi}$ & $\frac{8(1-\sqrt{5})}{5\pi}$ & \textbf{9.90} & \textbf{10.70} \\
\hline
\hline
\multirow{4}{*}{3}& {- -$\u \u \d$}, $\d \u \u$- -, -$\u \u$-$\d$, $\d$-$ \u \u$-  & $\frac{4}{5\pi}$ & $\frac{4(2+\sqrt{5})}{5\pi}$ & $\frac{4(2-\sqrt{5})}{5\pi}$ & 14.07 & 14.22\\
\cline{2-7}
& {$\d$-$\u \u \d$}, $\d \u \u$-$\d$ & 0 & $\frac{2(5+3\sqrt{5})}{5\pi}$ & $\frac{2(5-3\sqrt{5})}{5\pi}$ & 16.91 & 17.18\\
\cline{2-7}
& {- -$\u$-$\u$}, $\u$-$\u$- -, -$\u\u\d\u$, $\u\u\u\d$-, & \multirow{2}{*}{$\frac{8}{5\pi}$} & \multirow{2}{*}{$\frac{2(3-\sqrt{5})}{5\pi}$} & \multirow{2}{*}{$\frac{2(3+\sqrt{5})}{5\pi}$} & \multirow{2}{*}{\textbf{19.24}} & \multirow{2}{*}{\textbf{18.47}}\\
& -$\d\u\u\u$, $\u\d\u\u$- & & & & &\\
\hline
\hline
\multirow{5}{*}{4}& {- -$\u$- -}, $\u$-$\u$-$\d$, $\d$-$\u$-$\u$, & \multirow{3}{*}{$\frac{4}{5\pi}$} & \multirow{3}{*}{$\frac{8}{5\pi}$} & \multirow{3}{*}{$\frac{8}{5\pi}$} & \multirow{3}{*}{22.67} & \multirow{3}{*}{22.60}\\
& -$\u\u\d$-, $\u\u\u\d\d$, $\d\u\u\d\u$, & & & & &\\
& -$\d\u\u$-, $\u\d\u\u\d$, $\d\d\u\u\u$, & & & & &\\
\cline{2-7}
& {- -$\u$-$\d$}, $\d$-$\u$- -, -$\u\u\d\d$, $\d\u\u\d$-, & \multirow{2}{*}{0} & \multirow{2}{*}{$\frac{2(5+\sqrt{5})}{5\pi}$} & \multirow{2}{*}{$\frac{2(5-\sqrt{5})}{5\pi}$} & \multirow{2}{*}{24.58} & \multirow{2}{*}{24.68}\\
& -$\d\u\u\d$, $\d\d\u\u$-& & & & & \\
\cline{2-7}
& {$\d$-$ \u$-$\d$}, $\d\u \u\d\d$, $\d\d \u\u\d$ & $-\frac{4}{5\pi}$ & $\frac{4(3+\sqrt{5})}{5\pi}$ & $\frac{4(3-\sqrt{5})}{5\pi}$ & \textbf{25.84} & \textbf{26.03}\\
\hline
\hline
\multirow{6}{*}{5}& {$\d \d \u$-$\d$}, $\d $-$ \u\d\d$ & $-\frac{8}{5\pi}$ & $\frac{2(7+\sqrt{5})}{5\pi}$ & $\frac{2(7-\sqrt{5})}{5\pi}$ & 36.63 & 36.91\\
\cline{2-7}
& {- -$\u \d \d$}, $\d \d \u$- -, -$\d \u$-$\d$, $\d$-$ \u \d$- & $-\frac{4}{5\pi}$ & $\frac{12}{5\pi}$ & $\frac{12}{5\pi}$ & 37.24 & 37.52\\
\cline{2-7}
& {- -$\u \d$-}, -$\d \u$- -, $\u$-$\u \d\d$, $\u\d \u$-$\d$, & \multirow{2}{*}{0} & \multirow{2}{*}{$\frac{2(5-\sqrt{5})}{5\pi}$} & \multirow{2}{*}{$\frac{2(5+\sqrt{5})}{5\pi}$} & \multirow{2}{*}{38.07} & \multirow{2}{*}{38.35}\\
& $\d$-$\u \d\u$, $\d \d\u$-$\u$, &&&&&\\
\cline{2-7}
& {- -$\u \d \u$}, $\u \d \u$- -, -$\d \u$-$\u$, $\u$-$ \u \d$-  & $\frac{4}{5\pi}$ & $\frac{4(2-\sqrt{5})}{5\pi}$ & $\frac{4(2+\sqrt{5})}{5\pi}$ & 39.22 & 39.47\\
\cline{2-7}
& {$\u$-$\u \d\u$}, $\u \d \u$-$\u$ & $\frac{8}{5\pi}$ & $\frac{6(1-\sqrt{5})}{5\pi}$ & $\frac{6(1+\sqrt{5})}{5\pi}$ & \textbf{40.87} & \textbf{41.11}\\
\hline
\hline
\multirow{5}{*}{6}& {$\d \d \u \d \d$} & $-\frac{12}{5\pi}$ & $\frac{16}{5\pi}$ & $\frac{16}{5\pi}$ & 48.43 & 48.85\\
\cline{2-7}
& {$\d \d \u \d$-}, -$\d \u \d\d$ & $-\frac{8}{5\pi}$ & $\frac{2(7-\sqrt{5})}{5\pi}$ & $\frac{2(7+\sqrt{5})}{5\pi}$ & 50.43 & 50.86\\
\cline{2-7}
& {-$\d \u \d$-}, $\u\d \u \d\d$, $\d\d \u \d\u$ & $-\frac{4}{5\pi}$ & $\frac{4(3-\sqrt{5})}{5\pi}$ & $\frac{4(3+\sqrt{5})}{5\pi}$ & 52.78 & 53.25\\
\cline{2-7}
\cline{2-7}
& {$\u \d \u\d$-}, -$\d \u \d \u$ & 0 & $\frac{4(5-3\sqrt{5})}{5\pi}$ & $\frac{4(5+3\sqrt{5})}{5\pi}$ & 55.48 & 55.97\\
\cline{2-7}
& {$\u \d \u \d\u$} & $\frac{4}{5\pi}$ & $\frac{8(1-\sqrt{5})}{5\pi}$ & $\frac{8(1+\sqrt{5})}{5\pi}$ & \textbf{58.52} & \textbf{59.04}\\
\hline
\end{tabular}
\end{center}\vspace{-0.1in}
\end{table*}

For side wires, PDEs characterizing four wires with length $L$ are given by:

\begin{equation}
\frac{\partial^2}{\partial x^2}\mathbf{V}(x,t)
=\mathbf{RC}\frac{\partial}{\partial t}\mathbf{V}(x,t),
\label{Eq:PDE4}
\end{equation}
where $\mathbf{R}=\mbox{diag}\{r \; r \; r \; r\}$,
$\mathbf{V}(x,t)=[V_1(x,t)\;V_2(x,t)\;V_3(x,t)\;V_4(x,t)]^T$, and
$\mathbf{C}=c\left[ \begin{smallmatrix} 1+\lambda & -\lambda & 0 & 0 \\ -\lambda & 1+2\lambda & -\lambda & 0\\ 0 & -\lambda & 1+2\lambda & -\lambda \\ 0 & 0 & -\lambda & 1+\lambda \end{smallmatrix}\right]$.

The eigenvalues of $\mathbf{C}/c$ are given by $p_1=1$, $p_2=1+(2-\sqrt{2})\lambda$, $p_3=1+2\lambda$, and $p_4=1+(2+\sqrt{2})\lambda$. Their corresponding eigenvectors $\mathbf{e}_i$'s are given by $\mathbf{e}_1=[1\; 1\; 1\;1]^T$, $\mathbf{e}_2=[-1 \; (1-\sqrt{2}) \; -(1-\sqrt{2}) \; 1]^T$, $\mathbf{e}_3=[1 \; -1 \; -1 \; 1]^T$, and $\mathbf{e}_4=[-1 \; (1+\sqrt{2}) \; -(1+\sqrt{2}) \; 1]^T$, respectively.

By decoupling the PDEs in Eq.~(\ref{Eq:PDE4}), we have
\begin{equation}
\frac{\partial^2}{\partial x^2}U_i(x,t)=rcp_i\frac{\partial}{\partial t} U_i(x,t), \mbox{ for }i=1,2,3,4,\label{Eq:dePDE4}
\end{equation}

The expressions of wires 1 and 2 are given by $V_1(L,t)=\frac{1}{4}U_1(L,t) - \frac{2+\sqrt{2}}{8}U_2(L,t) + \frac{1}{4} U_3(L,t) -\frac{2-\sqrt{2}}{8}U_4(L,t)$ and $V_2(L,t)=\frac{1}{4}U_1(L,t) - \frac{\sqrt{2}}{8}U_2(L,t) - \frac{1}{4} U_3(L,t) +\frac{\sqrt{2}}{8}U_4(L,t)$, respectively.
Then the 50\% delays of wires 1 and 2 can be evaluated by solving $V_i(L,t)=0.5V_{dd}$ for $i=1,2$.

\subsection{Pattern Classification}
\label{sec:pat}
First, we consider the classification of transition patterns over five wires with respect to the delay of the middle wire (wire 3).
In this paper, we use ``$\uparrow$" to denote a transition from 0 to the supply voltage $V_{dd}$ (normalized to 1), ``-" no transition, and ``$\downarrow$" a transition from $V_{dd}$ to 0.
We first focus on patterns with a $\u$ transition on wire 3 in a five-wire bus and derive $V_3(L,t)$ for each pattern as described in Sec.~\ref{sec:dev}. There are $3^4=81$ different transition patterns, which can be partitioned into 25 subclasses according to the expressions of the output signals on wire 3: All transition patterns in each subclass have the same expression $V_3(L,t)$. The expressions of all 25 subclasses are shown in Tab.~\ref{tab:1}.
Then the expressions $V_3(L,t)$ of all patterns in the 25 subclasses are evaluated for their 50\% delays. By grouping subclasses with close delays into one class, we can divide the 81 transition patterns into seven classes $Ci$ for $i=0,1,\cdots, 6$ shown in Tab.~\ref{tab:1}. For all 25 subclasses, simulated delays are also provided in Tab.~\ref{tab:1}. For all seven classes, the difference between evaluated delay and simulated delay in Tab.~\ref{tab:1} is small.

All evaluations and simulations are based on a freePDK 45nm CMOS technology with 10 metal layers \cite{FreePDK45}. We assume that the top two metal layers, layers 9 and 10, are used for routing global interconnects, and that metal layer 8 is used as the ground layer.
An interconnect model in \cite{PTM} is used for parasitic extraction.
For a 5mm bus in the top metal layer, the key parasitics, resistance, ground capacitance, and coupling capacitance, are given by $R=68.75 \Omega$, $C_{gnd} = 41.32 fF$, and $C_{couple} = 505.68 fF$, respectively. The bus is modeled by a distributed RC model as shown in Fig.~\ref{fig:5dist} with 100 segments. The two important parameters used in our delay approximation are $\tau_0=0.5 R C_{gnd} = 1.42$ps and $\lambda=C_{couple} / C_{gnd} = 12.24$.
Since the crosstalk delay on the bus constitutes a major part of the whole delay, the delays introduced by buffers are ignored. We assume that ideal step signals are applied on the bus directly. The closed-form expressions are evaluated for 50\% delays via MATLAB and the simulation is done by HSPICE.

From Tab.~\ref{tab:1}, it can be easily verified that $C5$ and $C6$ are the same as $D3$ and $D4$ in~\cite{Sot01,Sot02}, respectively. That is, the middle three wires of the transition patterns in $C5$ ($C6$, respectively) constitute $D3$ ($D4$, respectively). The transition patterns in $D0$, $D1$, and $D2$ are divided into five classes $C0$---$C4$ in our classification with following relations, $C4 \subset D2$, $C3 \subset D1\cup D2$, $C2 \subset D0\cup D1$, $C1 \subset D0\cup D1 \cup D2$, and $C0 \subset D0\cup D1$.

Note that the coefficients $c_i$ for $i=0, 1, 2$ of the expression of wire 3 are independent of technology and determined by different patterns. For a given pattern, the coefficients $c_i$ are fixed and the delay is a function of $\tau_0$ and $\lambda$. Since the ratio $t/\tau_0$ appears in the exponent term, varying $\tau_0$ would scale delays in all classes. Thus, the classification does not depend on $\tau_0$. The coupling factor $\lambda$ could affect the delay differently.
In the following, we verify our classification for technology with different coupling factor, $\lambda=1,2,\cdots ,13$, and show the results in Fig.~\ref{fig:delay}.
Different classes are denoted by different line styles. Each class contains multiple lines, which represents a subclass. Patterns in each subclass have the same delay.
For $\lambda\ge 3$, the ranges of delays in all classes do not overlap. Also, the delay in each subclass increases linearly with $\lambda$. This implies that our classification is valid provided that the coupling factor $\lambda$\ is at least 3.

\begin{figure}[!tb]
\begin{minipage}[b]{1.0\linewidth}
  \centering
 \centerline{\epsfig{figure=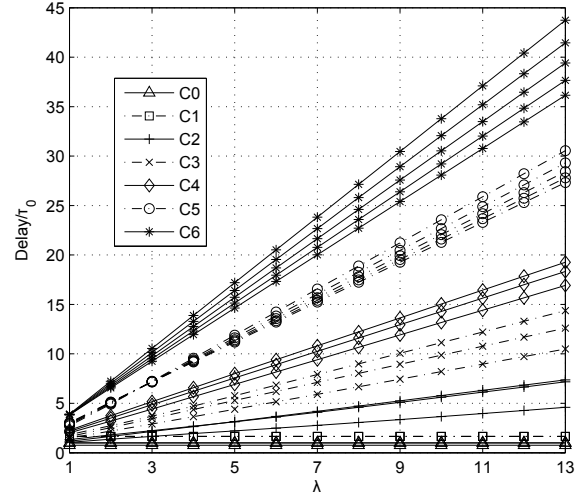,width=8.5cm}}
\end{minipage}
\caption{Delays of the middle wire for all patterns with respect to $\lambda$ in a five-wire bus ($\tau_0=1.42$ps).}
\label{fig:delay}
\end{figure}

\begin{figure}[!tb]
\begin{minipage}[b]{1.0\linewidth}
  \centering
 \centerline{\epsfig{figure=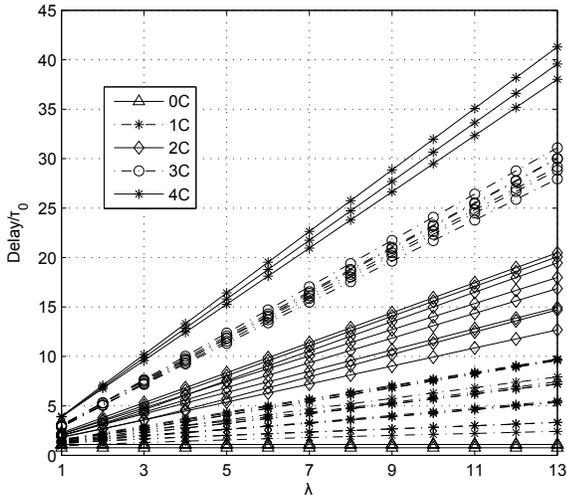,width=8.5cm}}
\end{minipage}
\caption{Delays of side wires for all patterns with respect to $\lambda$ in a four-wire bus ($\tau_0=1.42$ps).}
\label{fig:boundary}
\end{figure}

Then, we consider the classification of transition patterns over four wires with respect to the delays of the side wires. We classify patterns by considering the worst-case delays of wires 1 and 2, respectively. Note that the classification with respect to the delays of wires 4 and 5 would be the same by symmetry.
We first focus on patterns with a $\u$ transition on wire 2 in a four-wire bus. There are $3^3=27$ different transition patterns. As described in Sec.~\ref{sec:dev}, we first derive the expressions $V_2(L,t)$ of these 27 patterns shown in Tab.~\ref{tab:boundary2}. By evaluating these patterns for their 50\% delays, we group patterns with close delays into one class, and form 5 classes $jC$ for $j=0,1,2,3,4$ as shown in Tab.~\ref{tab:boundary2}.
Then, we focus on patterns with a $\u$ transition on wire 1. There are $3^3=27$ different transition patterns. As described in Sec.~\ref{sec:dev}, we first derive the expressions $V_1(L,t)$ of these 27 patterns shown in Tab.~\ref{tab:boundary1}. By evaluating these patterns for their 50\% delays, we group patterns with close delays into one class, and form 3 classes $jC$ for $j=0,1,2$ as shown in Tab.~\ref{tab:boundary1}.
When both wires 1 and 2 have transitions, the delay on wire 2 is larger than that of wire 1, which can be verified from Tabs.~\ref{tab:boundary2} and \ref{tab:boundary1}. In this case, we focus on the delay of wire 2.
When only wire 1 has transition, we focus on the delay of wire 1. The difference between evaluated delay and simulated delay is small as shown in Tabs.~\ref{tab:boundary2} and \ref{tab:boundary1} with one exception (the pattern $\u\u\d\u$ in $1C$ in Tab.~\ref{tab:boundary2}), which doesn't change our classification.

\begin{table*}[!htb]
\caption{Closed-form expressions for the output signals on wire 2 in a four-wire bus with evaluated and simulated 50\% delays ($\tau_0=1.42$~$\mathrm{ps}$,~$\tau=\frac{8}{\pi^2}\tau_0$, $\lambda=12.24$, $a_0=1$, $a_1=1+(2-\sqrt{2})\lambda$, $a_2=1+2\lambda$, and $a_3=1+(2+\sqrt{2})\lambda$ for all classes).}\label{tab:boundary2}
\begin{center}
\begin{tabular}{|c|l|c|c|c|c|c|c|}
\hline
\multirow{3}{*}{$jC$} & \multirow{3}{*}{Patterns} & \multicolumn{4}{|c|}{Closed-form expression for the output signal on wire 2} & \multirow{3}{*}{Evaluated delays (ps)} & \multirow{3}{*}{Sim. delay (ps)}\\
\cline{3-6} & & \multicolumn{4}{|c|}{$V_{dd}(1-c_0 e^{-\frac{t}{a_0\tau}} - c_1 e^{-\frac{t}{a_1\tau}} - c_2 e^{-\frac{t}{a_2\tau}} - c_3 e^{-\frac{t}{a_3\tau}})$} & &\\
\cline{3-6} & & $c_0$ & $c_1$ &  $c_2$ & $c_3$ & &\\
\hline
\hline
\multirow{4}{*}{0} & {$\u \u \u \u $} & $\frac{4}{\pi}$ & 0 & 0 & 0 & 1.08 & 1.18\\
\cline{2-8}
& $\u\u\u$- & $\frac{3}{\pi}$ & $\frac{\sqrt{2}}{2\pi}$ & $\frac{1}{\pi}$ & $-\frac{\sqrt{2}}{2\pi}$ & \textbf{1.55} & 1.61\\
\cline{2-8}
& $\u\u$-$\u$ & $\frac{3}{\pi}$ & $\frac{2-\sqrt{2}}{2\pi}$ & $-\frac{1}{\pi}$ & $-\frac{2+\sqrt{2}}{2\pi}$ & \textbf{1.55} & 1.62\\
\cline{2-8}
& -$\u\u\u$ & $\frac{3}{\pi}$ & $-\frac{\sqrt{2}}{2\pi}$ & $\frac{1}{\pi}$ & $\frac{\sqrt{2}}{2\pi}$ & \textbf{1.55} & \textbf{1.64}\\
\hline
\hline
\multirow{6}{*}{1}
& $\u\u\u\d$ & $\frac{2}{\pi}$ & $\frac{\sqrt{2}}{\pi}$ & $\frac{2}{\pi}$ & $-\frac{\sqrt{2}}{\pi}$ & 3.33 & 3.22\\
\cline{2-8}
& $\u\u$- - & $\frac{2}{\pi}$ & $\frac{1}{\pi}$ & 0 & $\frac{1}{\pi}$ & 4.54 & 3.48\\
\cline{2-8}
& -$\u\u$- & $\frac{2}{\pi}$ & 0 & $\frac{2}{\pi}$ & 0 & 7.21 & 5.15\\
\cline{2-8}
& $\u\u$-$\d$ & $\frac{1}{\pi}$ & $\frac{2+\sqrt{2}}{2\pi}$ & $\frac{1}{\pi}$ & $\frac{2-\sqrt{2}}{2\pi}$ & 9.70 & 9.38\\
\cline{2-8}
& $\u\u\d\u$ & $\frac{2}{\pi}$ & 0 & $\frac{2-\sqrt{2}}{2\pi}$ & $-\frac{2}{\pi}$ & 9.98 & 3.92\\
\cline{2-8}
& -$\u\u\d$ & $\frac{1}{\pi}$ & $\frac{\sqrt{2}}{2\pi}$ & $\frac{3}{\pi}$ & $\frac{-\sqrt{2}}{2\pi}$ & \textbf{12.89} & \textbf{13.03}\\
\hline
\hline
\multirow{8}{*}{2}
& $\u\u\d$- & $\frac{1}{\pi}$ & $\frac{4-\sqrt{2}}{2\pi}$ & $-\frac{1}{\pi}$ & $\frac{4+\sqrt{2}}{2\pi}$ & 17.02 & 16.05\\
\cline{2-8}
& -$\u$-$\u$ & $\frac{2}{\pi}$ & $\frac{1-\sqrt{2}}{\pi}$ & 0 & $\frac{1+\sqrt{2}}{\pi}$ & 19.67 & 18.79\\
\cline{2-8}
& $\u\u\d\d$ & 0 & $\frac{2}{\pi}$ & 0 & $\frac{2}{\pi}$ & 20.05 & 19.85\\
\cline{2-8}
& -$\u$- - & $\frac{1}{\pi}$ & $\frac{2-\sqrt{2}}{2\pi}$ & $\frac{1}{\pi}$ & $\frac{2+\sqrt{2}}{2\pi}$ & 22.59 & 22.48\\
\cline{2-8}
& -$\u$-$\d$ & 0 & $\frac{1}{\pi}$ & $\frac{2}{\pi}$ & $\frac{1}{\pi}$ & 24.12 & 24.22\\
\cline{2-8}
& $\d\u\u\u$ & $\frac{2}{\pi}$ & $-\frac{\sqrt{2}}{\pi}$ & $\frac{2}{\pi}$ & $\frac{\sqrt{2}}{\pi}$ & 26.02 & 26.06\\
\cline{2-8}
& $\d\u\u$- & $\frac{1}{\pi}$ & $-\frac{\sqrt{2}}{2\pi}$ & $\frac{3}{\pi}$ & $\frac{\sqrt{2}}{2\pi}$ & 26.89 & 27.06\\
\cline{2-8}
& $\d\u\u\d$ & 0 & 0 & $\frac{4}{\pi}$ & 0 & \textbf{27.45} & \textbf{27.68}\\
\hline
\hline
\multirow{6}{*}{3}
& -$\u\d\d$ & $-\frac{1}{\pi}$ & $\frac{4-\sqrt{2}}{2\pi}$ & $\frac{1}{\pi}$ & $\frac{4+\sqrt{2}}{2\pi}$ & 37.44 & 37.74\\
\cline{2-8}
& -$\u\d$- & 0 & $\frac{2-\sqrt{2}}{\pi}$ & 0 & $\frac{2+\sqrt{2}}{\pi}$ & 38.61 & 38.89\\
\cline{2-8}
& $\d\u$-$\d$ & $-\frac{1}{\pi}$ & $\frac{2-\sqrt{2}}{2\pi}$ & $\frac{3}{\pi}$ & $\frac{2+\sqrt{2}}{2\pi}$ & 39.06 & 39.40\\
\cline{2-8}
& -$\u\d\u$ & $\frac{1}{\pi}$ & $\frac{4-\sqrt{2}}{2\pi}$ & $-\frac{1}{\pi}$ & $\frac{4+\sqrt{2}}{2\pi}$ & 40.12 & 40.39\\
\cline{2-8}
& $\d\u$- - & 0 & $\frac{1-\sqrt{2}}{\pi}$ & $\frac{2}{\pi}$ & $\frac{1+\sqrt{2}}{\pi}$ & 40.21 & 40.55\\
\cline{2-8}
& $\d\u$-$\u$ & $\frac{1}{\pi}$ & $\frac{2-3\sqrt{2}}{2\pi}$ & $\frac{1}{\pi}$ & $\frac{2+3\sqrt{2}}{2\pi}$ & \textbf{41.63} & \textbf{41.98}\\
\hline
\hline
\multirow{3}{*}{4}
& $\d\u\d\d$ & $-\frac{2}{\pi}$ & $\frac{2-\sqrt{2}}{\pi}$ & $\frac{2}{\pi}$ & $\frac{2+\sqrt{2}}{\pi}$ & 50.92 & 51.36\\
\cline{2-8}
& $\d\u\d$- & $-\frac{1}{\pi}$ & $\frac{4-3\sqrt{2}}{2\pi}$ & $\frac{1}{\pi}$ & $\frac{4+3\sqrt{2}}{2\pi}$ & 52.99 & 53.44\\
\cline{2-8}
& $\d\u\d\u$ & 0 & $\frac{2-2\sqrt{2}}{\pi}$ & 0 & $\frac{2+2\sqrt{2}}{\pi}$ & \textbf{55.28} & \textbf{55.79}\\
\hline
\end{tabular}
\end{center}
\end{table*}

\begin{table*}[!htb]
\caption{Closed-form expressions for the output signals on wire 1 in a four-wire bus with evaluated and simulated 50\% delays ($\tau_0=1.42$~$\mathrm{ps}$,~$\tau=\frac{8}{\pi^2}\tau_0$, $\lambda=12.24$, $a_0=1$, $a_1=1+(2-\sqrt{2})\lambda$, $a_2=1+2\lambda$, and $a_3=1+(2+\sqrt{2})\lambda$ for all classes).}\label{tab:boundary1}
\begin{center}
\begin{tabular}{|c|l|c|c|c|c|c|c|}
\hline
\multirow{3}{*}{$jC$} & \multirow{3}{*}{Patterns} & \multicolumn{4}{|c|}{Closed-form expression for the output signal on wire 1} & \multirow{3}{*}{Evaluated delays (ps)} & \multirow{3}{*}{Sim. delay (ps)}\\
\cline{3-6} & & \multicolumn{4}{|c|}{$V_{dd}(1-c_0 e^{-\frac{t}{a_0\tau}} - c_1 e^{-\frac{t}{a_1\tau}} - c_2 e^{-\frac{t}{a_2\tau}} - c_3 e^{-\frac{t}{a_3\tau}})$} & &\\
\cline{3-6} & & $c_0$ & $c_1$ &  $c_2$ & $c_3$ & &\\
\hline
\hline
\multirow{4}{*}{0}
& $\u\u\u\u$ & $\frac{4}{\pi}$ & 0 & 0 & 0 & 1.08 & 1.18\\
\cline{2-8}
& $\u\u\u$- & $\frac{3}{\pi}$ & $-\frac{2+\sqrt{2}}{2\pi}$ & $-\frac{1}{\pi}$ & $\frac{2-\sqrt{2}}{2\pi}$ & \textbf{1.55} & 1.59\\
\cline{2-8}
& $\u\u$-$\u$ & $\frac{3}{\pi}$ & $\frac{\sqrt{2}}{2\pi}$ & $\frac{1}{\pi}$ & $-\frac{\sqrt{2}}{2\pi}$ & \textbf{1.55} & 1.61\\
\cline{2-8}
& $\u$-$\u\u$ & $\frac{3}{\pi}$ & $-\frac{\sqrt{2}}{2\pi}$ & $\frac{1}{\pi}$ & $\frac{\sqrt{2}}{2\pi}$ & \textbf{1.55} & \textbf{1.64}\\
\hline
\hline
\multirow{14}{*}{1}
& $\u\u\u\d$ & $\frac{2}{\pi}$ & $\frac{2+\sqrt{2}}{\pi}$ & $-\frac{2}{\pi}$ & $\frac{2-\sqrt{2}}{\pi}$ & 2.50 & 2.70\\
\cline{2-8}
& $\u\u$- - & $\frac{2}{\pi}$ & $\frac{1+\sqrt{2}}{\pi}$ & 0 & $\frac{1-\sqrt{2}}{\pi}$ & 2.83 & 2.90\\
\cline{2-8}
& $\u\u\d\u$ & $\frac{2}{\pi}$ & $\frac{\sqrt{2}}{\pi}$ & $\frac{2}{\pi}$ & $-\frac{\sqrt{2}}{\pi}$ & 3.33 & 3.20\\
\cline{2-8}
& $\u\u$-$\d$ & $\frac{1}{\pi}$ & $\frac{4+3\sqrt{2}}{2\pi}$ & $-\frac{1}{\pi}$ & $\frac{4-3\sqrt{2}}{2\pi}$ & 4.65 & 4.99\\
\cline{2-8}
& $\u$-$\u$- & $\frac{2}{\pi}$ & $\frac{1}{2\pi}$ & 0 & $\frac{1}{2\pi}$ & 4.54 & 3.49\\
\cline{2-8}
& $\u\u\d$- & $\frac{1}{\pi}$ & $\frac{2+3\sqrt{2}}{2\pi}$ & $\frac{1}{\pi}$ & $\frac{2-3\sqrt{2}}{2\pi}$ & 5.53 & 5.88\\
\cline{2-8}
& $\u\u\d\d$ & 0 & $\frac{2+2\sqrt{2}}{\pi}$ & 0 & $\frac{2-2\sqrt{2}}{\pi}$ & 7.03 & 7.39\\
\cline{2-8}
& $\u$- -$\u$ & $\frac{2}{\pi}$ & 0 & $\frac{2}{\pi}$ & 0 & 7.21 & 5.15\\
\cline{2-8}
& $\u$-$\u\d$ & $\frac{1}{\pi}$ & $\frac{4+\sqrt{2}}{2\pi}$ & $-\frac{1}{\pi}$ & $\frac{4-\sqrt{2}}{2\pi}$ & 7.41 & 6.89\\
\cline{2-8}
& $\u$- - - & $\frac{1}{\pi}$ & $\frac{2+\sqrt{2}}{2\pi}$ & $\frac{1}{\pi}$ & $\frac{2-\sqrt{2}}{2\pi}$ & 9.70 & 9.35\\
\cline{2-8}
& $\u$- -$\d$ & 0 & $\frac{2+\sqrt{2}}{\pi}$ & 0 & $\frac{2-\sqrt{2}}{\pi}$ & 10.68 & 10.54\\
\cline{2-8}
& $\u$-$\d\u$ & $\frac{1}{\pi}$ & $\frac{\sqrt{2}}{2\pi}$ & $\frac{3}{\pi}$ & $\frac{-\sqrt{2}}{2\pi}$ & 12.89 & 13.03\\
\cline{2-8}
& $\u$-$\d$- & 0 & $\frac{2+2\sqrt{2}}{2\pi}$ & $\frac{2}{\pi}$ & $\frac{2-2\sqrt{2}}{2\pi}$ & 13.03 & 13.14\\
\cline{2-8}
& $\u$-$\d\d$ & $-\frac{1}{\pi}$ & $\frac{4+3\sqrt{2}}{2\pi}$ & $\frac{1}{\pi}$ & $\frac{4-3\sqrt{2}}{2\pi}$ & \textbf{13.11} & \textbf{13.21}\\
\hline
\hline
\multirow{9}{*}{2}
& $\u\d\u\d$ & 0 & $\frac{2}{\pi}$ & 0 & $\frac{2}{\pi}$ & 20.05 & 19.85\\
\cline{2-8}
& $\u\d$-$\d$ & $-\frac{1}{\pi}$ & $\frac{4+\sqrt{2}}{2\pi}$ & $\frac{1}{\pi}$ & $\frac{4-\sqrt{2}}{2\pi}$ & 21.86 & 21.91\\
\cline{2-8}
& $\u\d\u$- & $\frac{1}{\pi}$ & $\frac{2-\sqrt{2}}{2\pi}$ & $\frac{1}{\pi}$ & $\frac{2+\sqrt{2}}{2\pi}$ & 22.59 & 22.48\\
\cline{2-8}
& $\u\d\d\d$ & $-\frac{2}{\pi}$ & $\frac{2+\sqrt{2}}{\pi}$ & $\frac{2}{\pi}$ & $\frac{2-\sqrt{2}}{\pi}$ & 23.10 & 23.23\\
\cline{2-8}
& $\u\d$- - & 0 & $\frac{1}{\pi}$ & $\frac{2}{\pi}$ & $\frac{1}{\pi}$ & 24.12 & 24.22\\
\cline{2-8}
& $\u\d\d$- & $-\frac{1}{\pi}$ & $\frac{2+\sqrt{2}}{2\pi}$ & $\frac{3}{\pi}$ & $\frac{2-\sqrt{2}}{2\pi}$ & 25.10 & 25.30\\
\cline{2-8}
& $\u\d\u\u$ & $\frac{2}{\pi}$ & $\frac{-\sqrt{2}}{\pi}$ & $\frac{2}{\pi}$ & $\frac{\sqrt{2}}{\pi}$ & 26.02 & 26.06\\
\cline{2-8}
& $\u\d$-$\u$ & $\frac{1}{\pi}$ & $-\frac{\sqrt{2}}{2\pi}$ & $\frac{3}{\pi}$ & $\frac{\sqrt{2}}{2\pi}$ & 26.89 & 27.06\\
\cline{2-8}
& $\u\d\d\u$ & 0 & 0 & $\frac{4}{\pi}$ & 0 & \textbf{27.45} & \textbf{27.68}\\
\hline
\end{tabular}
\end{center}\vspace{-0.1in}
\end{table*}

From Tabs.~\ref{tab:boundary2} and \ref{tab:boundary1}, the classes $3C$ and $4C$ of our classification are exactly the same as $D3$ and $D4$ in \cite{Sot01, Sot02}, respectively. The class $1C$ and $2C$ of our classification are subsets of $D1$ and $D2$ in \cite{Sot01,Sot02}, respectively. The class $0C$ is a subset of $D0\cup D1$ in \cite{Sot01,Sot02}.

Similar to the classification of middle wires, we conclude that the classification on side wires does not depend on $\tau_0$. To verify our classification for technology with different coupling effects, we consider coupling factor $\lambda=1,2,\cdots ,13$, and show the results in Fig.~\ref{fig:boundary}.
Each class contains multiple lines, each of which represents a pattern in Tabs.~\ref{tab:boundary2} and \ref{tab:boundary1}.
For $\lambda\ge 1$, the ranges of delays in all classes do not overlap. Also, the delay in each subclass increases linearly with $\lambda$. This implies that our classification on side wires is valid provided that the coupling factor $\lambda$\ is at least 1.

In addition to being a finer classification, the new classification has no overlapping delays among different classes. Fig.~\ref{fig:mitclass} compares the simulated delays of different classes based on the classification in~\cite{Sot01,Sot02} and our new classification. In Fig.~\ref{fig:mitclass}, the grey bars identify the minimum and maximum simulated delays in every class. Note that only two extremes are important, and not all delay values in the grey bars are achievable by some transition patterns.
In Fig.~\ref{fig:mitclass}(a), the thick line segments denote the upper bounds for delay of each class based on Eq.~(\ref{Eq:1}). The upper bounds by the model in~\cite{Sot01,Sot02} overestimate the delays of $D1$ through $D4$ and underestimate the delay of $D0$. As shown in Fig.~\ref{fig:mitclass}(a), the actual delays in $D0$, $D1$, and $D2$ overlap with each other. Some patterns with smaller delays have potential to transmit information at a higher speed, but are categorized into a class with a larger delay bound. Thus, the classification by the model in~\cite{Sot01,Sot02} does not result in effective crosstalk avoidance codes. In contrast, the delays of different classes in our new classification do not overlap as shown in Fig.~\ref{fig:mitclass}(b), 4(c), and 4(d). By classifying patterns this way, we have a more accurate control of delays for transition patterns.

\begin{figure}[!tb]
\begin{minipage}[b]{1.0\linewidth}
  \centering
 \centerline{\epsfig{figure=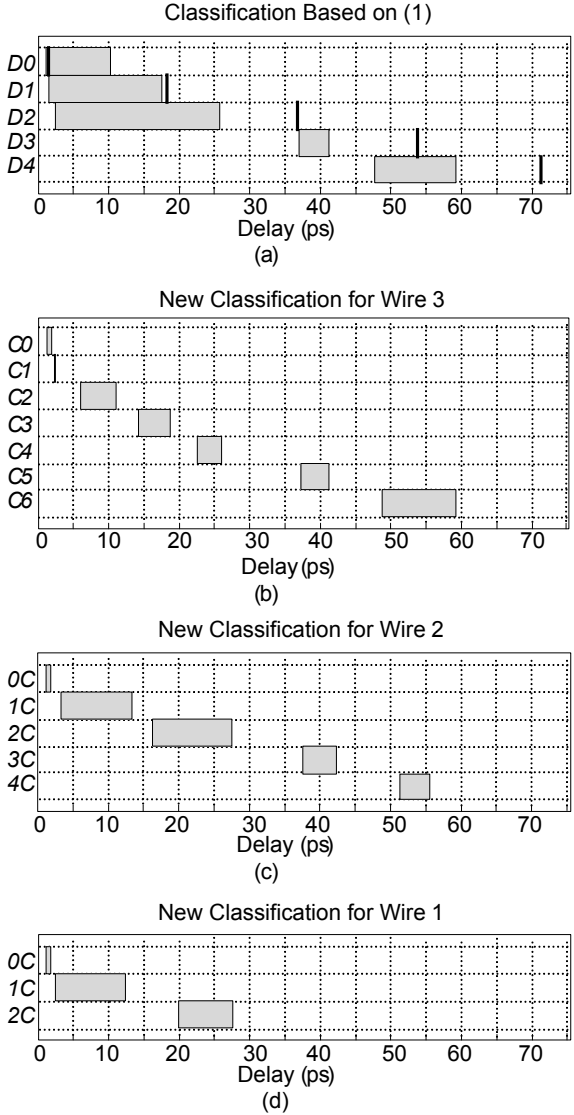,width=7.5cm}}
\end{minipage}
\caption{Simulated delays of different classes of transition patterns using (a) Classification based on (\ref{Eq:1}); (b) Classification with respect to the delay of the middle wire in a five-wire bus; (c) Classification with respect to the delay of wire 2 in a four-wire bus; (d) Classification with respect to the delay of wire 1 in a four-wire bus ($\lambda=12.24$ and $\tau_0=1.42$ps).}
\label{fig:mitclass}\vspace{-0.1in}
\end{figure}

\section{NEW MEMORYLESS CROSSTALK AVOIDANCE CODES}\label{sec:cacs}

\subsection{Previous CAC Design}
CACs reduce the crosstalk delay for on-chip global interconnects by encoding a $k$-bit data word $(x_1 x_2 \cdots x_k)$ into an $n$-bit ($n>k$) codeword $(c_1 c_2 \cdots c_n)$.
Two kinds of CACs, CACs with memory and memoryless CACs, have been investigated in the literature. CACs with memory, as shown in Fig.~\ref{fig:codec}(a), need to store all codebooks corresponding to different codewords $(c_1 c_2 \cdots c_n)$, since the encoding depends on the data word $(x_1 x_2 \cdots x_k)$ as well as the preceding codeword.
In contrast, memoryless CACs, as shown in Fig.~\ref{fig:codec}(b), require a single codebook to generate codewords for transmission, because the encoding depends on the data word only.
Hence, memoryless CACs are simpler to implement than CACs with memory.
We focus on memoryless CACs in this paper.

The codebook of a memoryless CAC satisfies the property that each codeword must be able to transition to every other codeword in the codebook with a delay less than the requirement. Most memoryless CACs in the literature are based on the model in~\cite{Sot01,Sot02}. The key idea is to eliminate undesirable patterns for transmission.
Existing memoryless CACs include OLCs, FPCs, FTCs, and FOCs~\cite{Dua01,Dua04,Vic01,Sri07}, which achieve a worst-case delay of $(1+\lambda)\tau_0$, $(1+2\lambda)\tau_0$, $(1+2\lambda)\tau_0$, and $(1+3\lambda)\tau_0$, respectively. As mentioned above, the scheme that was proposed to achieve a worst-case delay of $\tau_0$ is invalid since the model in \cite{Sot01,Sot02} underestimates the delays for $0C$. Thus, OLCs achieve the smallest worst-case delay $(1+\lambda)\tau_0$ among existing CACs.

There exist several methods to obtain a memoryless codebook based on pattern pruning, transition pruning, or recursive construction. The pattern pruning technique is quite straight forward, and gives a codebook with a smaller worst-case delay by eliminating some patterns. For example, FOCs cannot have both 010 and 101 patterns around any bit position, and FPCs are free of 010 and 101 patterns \cite{Sri07}.
The transition pruning technique \cite{Vic01} is based on graph theory.
This method first builds a transition graph with all possible codewords as nodes and all valid transitions as edges, and then finds a maximum clique. A clique is defined as a subgraph where every pair of nodes are connected with an edge. A maximum clique is defined as a clique of the largest possible size in a given graph. Since every pair of nodes is connected, a maximum clique in this graph constitutes a memoryless codebook with the largest size. The codebook generation method is based on exhaustive search. Although it is easy to get a maximum clique from a transition graph with a small $n$, the complexity increases rapidly with $n$. This is because the number of edges in an $n$-bit transition graph is upper bounded by $2^{n-1}(2^n-1)$, which increases exponentially with $n$.  In fact, it is an NP problem to find a maximum clique for given constraints \cite{Gar79}.
The recursive technique constructs an $(n+1)$-bit codebook from an $n$-bit codebook \cite{Dua01,Dua04}. Since for a small $n$, a largest codebook can be obtained easily via the second method, a codebook for an $n$-wire bus can be constructed recursively.

\begin{figure*}[!t]
\begin{minipage}[b]{1.0\linewidth}
  \centering
 \centerline{\epsfig{figure=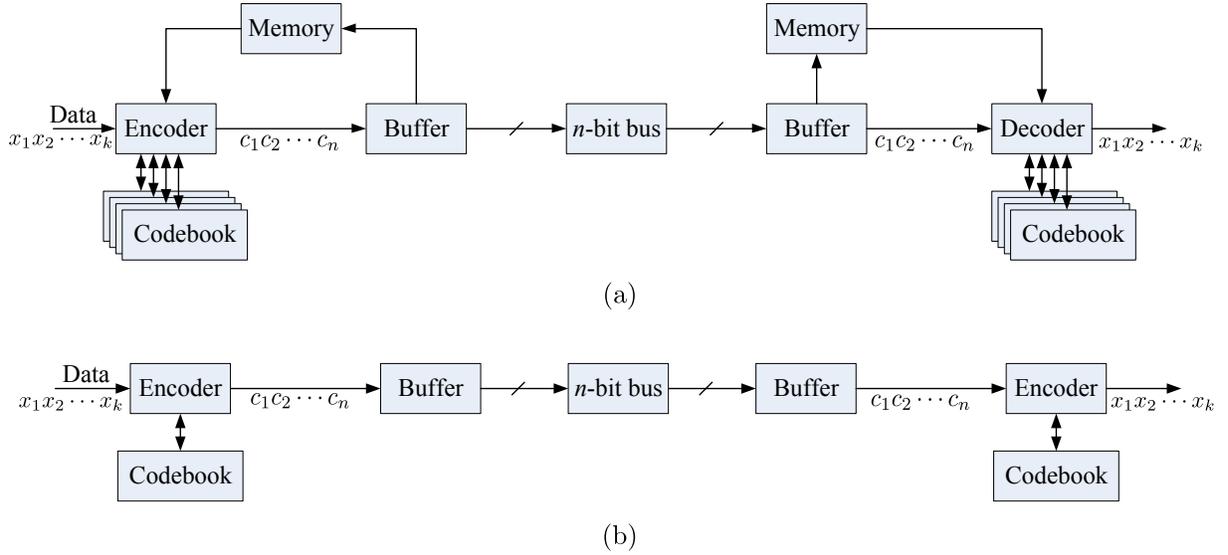,width=16cm}}
\end{minipage}
\caption{System model for (a) CACs with memory; (b) Memoryless CACs.}
\label{fig:codec}
\end{figure*}

\vspace{-0.1in}
\subsection{CAC Design with New Classification}
Since our classification of patterns is different from that in \cite{Sot01,Sot02}, the CAC designs should be reconsidered with our new classification. In the following, we first introduce a recursive method for codebook construction under different constraints, and then derive the size of codebooks.

In our work, we use the recursive method to obtain a memoryless codebook for the following two reasons. First, it is complex to apply the pattern pruning technique, since our new classification is based on transitions over five wires, and it is not clear which patterns have larger worst-case delays and should be removed. Second, it is hard to find a maximum clique for a transition graph with a large $n$.
In our method, we first start with a 5-bit codebook, obtained by searching for maximum cliques in a five-wire bus, and then build an $(n+1)$-bit codebook by appending '0' and '1' to codewords of an $n$-bit codebook while satisfying delay constraints.

Our new classifications partition patterns over five adjacent wires into seven classes, $C0$ to $C6$, and patterns over four adjacent wires into five classes, $0C$ to $4C$. Similar to the CAC design based on the model in~\cite{Sot01,Sot02}, the new classifications are conducive to the design of CACs by eliminating undesirable transition patterns with large worst-case delays.

To get valid 5-bit codebooks, we first assume the allowed patterns are from $C0$ to $Ci$ for $i=0,1,\cdots, 6$ in our classification for middle wires. Then, for the side wires, we assume patterns are from $0C$ to $jC$ based on the classification for side wires. Under these two assumptions, there are many configurations of constraints, which are referred as $(Ci,jC)$, where $i\in \{0, 1, \cdots, 6\}$ and $j \in \{0, 1, \cdots, 4\}$.

Since the worst-case delay of a bus is determined by the largest delays among all wires, for an $n$-bit ($n\ge 5$) bus under $(Ci,jC)$ we require that the worst-case delays on middle wires and side wires are close enough. By our classifications, we find $0C$ is close to $C0$, $1C$ close to $C2$ and $C3$, $2C$ close to $C4$, $3C$ close to $C5$, and $4C$ close to $C6$.
Hence, among all configurations of constraints $(Ci,jC)$, we only focus on $(C0,0C)$, $(C2,1C)$, $(C3,1C)$, $(C4,2C)$, $(C5,3C)$, and $(C6,4C)$. When $n\le 4$, the constraint $Ci$ cannot be enforced. Hence, the constraint $(Ci,jC)$ reduces to $jC$. The constraint $(C0,0C)$ appears to be too restrictive, and hence we do not investigate it in this paper.
The last configuration $(C6,4C)$ is trivial, since it allows arbitrary transitions.

In the following, we propose a scheme for finding an $n$-bit codebook $C_{(Ci,jC)}(n)$. For simplicity, we denote $C_{(Ci,jC)}(n)$ as $C(n)$ when there is no ambiguity about the constraint. First, for a five-wire bus under constraint $(Ci,jC)$, a pattern transition graph is obtained. We search the graph for the largest 5-bit codebooks. One or two 5-bit codebooks of maximum sizes exist for each constraint in Tab.~\ref{tab:cb}, where we denote an $n$-bit binary codeword $(c_1 c_2 \cdots c_n)$ as a decimal number $\sum_{i=1}^{n} c_i 2^{n-i}$ for simplicity.
In~\cite{Vic01}, a bit boundary in a set of codewords is said to be $01$-type if only codewords with 00, 01, and 11 are allowed across that boundary, and a bit boundary is said to be $10$-type when only codewords with 00, 10, and 11 are allowed across that boundary. It is shown that the largest clique for a given constraint has alternating boundary types. Thus, there are two largest cliques. Similarly, from Tab.~\ref{tab:cb}, we conjecture that the largest codebooks have alternating constraints, $C^0_5$ and $C^1_5$, for every five consecutive wires.
For constraint $(C4,2C)$, only one maximum 5-bit codebook exists. We assume $C^1_5$ is the same as $C^0_5$ for constraint $(C4,2C)$. Since we have two types of constraints, two largest codebooks for each constraint can be obtained, except for ($C4,2C$), where the two codebooks are the same.
Then we apply Alg.~\ref{alg:generic} to obtain $C(n)$.
In the initialization, we pick a 5-bit codebook $C_5 = C_5^0$. Then, the algorithm recursively appends one bit to the codewords in the codebook in each iteration. For $\mathbf{c}_k=(c_1c_2\cdots c_k)$, the appended bit $x$ needs to satisfy that the last five bits ($c_{k-3}c_{k-2}c_{k-1}c_k x$) form a codeword in $C_5^s$, which alternates between $C_5^0$ and $C_5^1$. If we pick the other 5-bit codebook $C_5=C_5^1$, we would obtain another codebook.

\begin{table*}[!htp]
\caption{Largest 5-bit codebook(s) under constraint $(Ci,jC)$.}\label{tab:cb}
\begin{center}
\begin{tabular}{|c|c|c|}
\hline
Constraint & $C_5^0$ & $C_5^1$\\
\hline
\multirow{2}{*}{$(C5,3C)$} & \{0, 1, 2, 3, 6, 7, 8, 9, 10, 11, 12, 14, & \{0, 1, 3, 4, 5, 6, 7, 12, 13, 14, 15, 16,\\
& 15, 16, 17, 18, 19, 24, 25, 26, 27, 28, 30, 31\} &  17, 19, 20, 21, 22, 23, 24, 25, 28, 29, 30, 31\}\\
\hline
$(C4,2C)$ & \{0, 1, 3, 6, 7, 12, 14, 15, 16, 17, 19, 24, 25, 28, 30, 31\} & \\
\hline
$(C3,1C)$ & \{0, 3, 14, 15, 24, 30, 31\} & \{0, 1, 7, 16, 17, 28, 31\}\\
\hline
$(C2,1C)$ & \{0, 3, 15, 24, 30, 31\} & \{0, 1, 7, 16, 28, 31\}\\
\hline
\end{tabular}
\end{center}
\end{table*}

\begin{table*}[!htp]
\caption{Expansion matrix for $(C3,1C)$, $(C4,2C)$, and $(C5,3C)$.}\label{tab:expmtx}
\begin{center}
\begin{tabular}{c c c}
$\mathbf{D}_{(C3,1C)} = \left[ \begin{smallmatrix}  0&0&0&0&0&1&1 \\ 0&0&0&0&1&0&0 \\ 0&1&0&0&0&0&0 \\ 1&0&0&0&0&0&0 \\ 0&0&1&1&0&0&0 \\ 0&1&0&0&0&0&0 \\ 1&0&0&0&0&0&0 \end{smallmatrix} \right]$, &
$ \mathbf{D}_{(C4,2C)} = \left[ \begin{smallmatrix}  0&0&0&0&0&0&0&0&0&0&0&0&0&0&1&1 \\ 0&0&0&0&0&0&0&0&0&0&0&0&0&1&0&0\\ 0&0&0&0&0&0&0&0&0&0&0&1&1&0&0&0 \\ 0&0&0&0&0&0&0&0&0&0&1&0&0&0&0&0 \\ 0&0&0&0&0&0&0&0&1&1&0&0&0&0&0&0 \\ 0&0&0&1&1&0&0&0&0&0&0&0&0&0&0&0 \\ 0&0&1&0&0&0&0&0&0&0&0&0&0&0&0&0 \\ 1&1&0&0&0&0&0&0&0&0&0&0&0&0&0&0 \\ 0&0&0&0&0&0&0&0&0&0&0&0&0&0&1&1 \\ 0&0&0&0&0&0&0&0&0&0&0&0&0&1&0&0 \\ 0&0&0&0&0&0&0&0&0&0&0&1&1&0&0&0 \\ 0&0&0&0&0&0&1&1&0&0&0&0&0&0&0&0 \\ 0&0&0&0&0&1&0&0&0&0&0&0&0&0&0&0 \\ 0&0&0&1&1&0&0&0&0&0&0&0&0&0&0&0 \\  0&0&1&0&0&0&0&0&0&0&0&0&0&0&0&0 \\ 1&1&0&0&0&0&0&0&0&0&0&0&0&0&0&0  \end{smallmatrix} \right]$, &
$ \mathbf{D}_{(C5,3C)} = \left[ \begin{smallmatrix}  0&0&0&0&0&0&0&0&0&0&0&0&0&0&0&0&0&0&0&0&0&0&1&1 \\ 0&0&0&0&0&0&0&0&0&0&0&0&0&0&0&0&0&0&0&0&0&1&0&0 \\ 0&0&0&0&0&0&0&0&0&0&0&0&0&0&0&0&0&0&0&1&1&0&0&0 \\ 0&0&0&0&0&0&0&0&0&0&0&0&0&0&0&0&0&1&1&0&0&0&0&0 \\ 0&0&0&0&0&0&0&0&0&0&0&0&0&0&0&1&1&0&0&0&0&0&0&0 \\ 0&0&0&0&0&0&0&0&0&0&0&0&0&1&1&0&0&0&0&0&0&0&0&0 \\ 0&0&0&0&0&0&0&0&0&0&0&1&1&0&0&0&0&0&0&0&0&0&0&0 \\ 0&0&0&0&0&0&0&0&0&0&1&0&0&0&0&0&0&0&0&0&0&0&0&0 \\ 0&0&0&0&0&0&0&0&1&1&0&0&0&0&0&0&0&0&0&0&0&0&0&0 \\ 0&0&0&0&0&0&1&1&0&0&0&0&0&0&0&0&0&0&0&0&0&0&0&0 \\ 0&0&0&0&1&1&0&0&0&0&0&0&0&0&0&0&0&0&0&0&0&0&0&0 \\ 0&0&1&1&0&0&0&0&0&0&0&0&0&0&0&0&0&0&0&0&0&0&0&0 \\ 1&1&0&0&0&0&0&0&0&0&0&0&0&0&0&0&0&0&0&0&0&0&0&0 \\ 0&0&0&0&0&0&0&0&0&0&0&0&0&0&0&0&0&0&0&0&0&0&1&1 \\ 0&0&0&0&0&0&0&0&0&0&0&0&0&0&0&0&0&0&0&0&0&1&0&0 \\ 0&0&0&0&0&0&0&0&0&0&0&0&0&0&0&0&0&0&0&1&1&0&0&0 \\ 0&0&0&0&0&0&0&0&0&0&0&0&0&0&0&0&0&1&1&0&0&0&0&0 \\ 0&0&0&0&0&0&0&0&0&0&0&1&1&0&0&0&0&0&0&0&0&0&0&0 \\ 0&0&0&0&0&0&0&0&0&0&1&0&0&0&0&0&0&0&0&0&0&0&0&0 \\ 0&0&0&0&0&0&0&0&1&1&0&0&0&0&0&0&0&0&0&0&0&0&0&0 \\ 0&0&0&0&0&0&1&1&0&0&0&0&0&0&0&0&0&0&0&0&0&0&0&0 \\ 0&0&0&0&1&1&0&0&0&0&0&0&0&0&0&0&0&0&0&0&0&0&0&0 \\ 0&0&1&1&0&0&0&0&0&0&0&0&0&0&0&0&0&0&0&0&0&0&0&0 \\ 1&1&0&0&0&0&0&0&0&0&0&0&0&0&0&0&0&0&0&0&0&0&0&0  \end{smallmatrix} \right]$.
\end{tabular}
\end{center}
\end{table*}

\begin{algorithm}[!htb]
  \caption{Codebook design under $(Ci,jC)$}
  \begin{algorithmic}
    \REQUIRE $C^0_5$, $C^1_5$, $n$;
    \STATE \textbf{Initialize}: $k=5$, $C_5=C^0_5$, $s=1$;
    \WHILE{$k \le n-1$}
    \FOR{$\forall \mathbf{c}_k =(c_1c_2\cdots c_k) \in C(k)$}
    \IF{$(c_{k-3}c_{k-2}c_{k-1}c_k 0) \in C^s_5$}
    \STATE append 0 to $\mathbf{c}_k$ and add the new codeword to $C(k+1)$;
    \ELSIF{$(c_{k-3}c_{k-2}c_{k-1}c_k 1) \in C^s_5$}
    \STATE append 1 to $\mathbf{c}_k$ and add the new codeword to $C(k+1)$;
    \ENDIF
    \ENDFOR
    \STATE $s=1-s$;
    \STATE $k=k+1$;
    \ENDWHILE
    \RETURN $C(n)$.
  \end{algorithmic}
  \label{alg:generic}
\end{algorithm}


The recursive construction allows us to derive the size of the codebooks. Let $\mathbf{V}_{(Ci,jC)}$ be an all-one  $m$-dimensional row vector ($m=|C_5^0|$) under constraint $(Ci,jC)$. Let $\mathbf{c}^s_k$ be a $k$-bit codeword with last five consecutive bits $(c_{k-4}c_{k-3}c_{k-2}c_{k-1}c_k) \in C^s_5$ for $s=0$ or $1$.
If a $0$ or $1$ can be appended to $\mathbf{c}^s_k$ to form a $(k+1)$-bit codeword whose last five bits $(c_{k-3}c_{k-2}c_{k-1}c_{k}c_{k+1}) \in C^{1-s}_5$, such an expansion is called a valid expansion. Otherwise, it is called an invalid expansion.
An expansion matrix is denoted as a $m \times m$ matrix $\mathbf{D}^s_{(Ci,jC)}$, where $\mathbf{D}^s_{(Ci,jC)}(i,j)=0$ denotes an invalid expansion and $\mathbf{D}_{(Ci,jC)}^s(i,j)=1$ a valid expansion from the $i$-th codeword in $C_5^s$ to the $j$-th codeword in $C_5^{1-s}$ under constraint $(Ci,jC)$.
Each row of $\mathbf{D}^s_{(Ci,jC)}$ has at most two ones, since each $k$-bit codeword can be appended to form at most two $(k+1)$-bit codewords whose last five bits satisfy the appropriate constraints.
Let $\mathbf{Y}$ be an $m\times m$ anti-diagonal matrix with all ones. Due to symmetry between $C_5^0$ and $C_5^1$, $\mathbf{D}^0$ and $\mathbf{D}^1$ satisfy $\mathbf{D}_{(Ci,jC)}^1 = \mathbf{Y} \mathbf{D}_{(Ci,jC)}^0 \mathbf{Y}$. Define $\mathbf{D}_{(Ci,jC)} = \mathbf{D}^0_{(Ci,jC)} \mathbf{Y} = \mathbf{Y} \mathbf{D}^1_{(Ci,jC)}$. We denote $\mathbf{V}_{(Ci,jC)}$ and $\mathbf{D}_{(Ci,jC)}$ as $\mathbf{V}$ and $\mathbf{D}$, respectively, when there is no ambiguity about the constraint.
Then, for $n\ge 5$, the number of codewords in an $n$-bit bus is equal to counting the valid transitions and is given by
\begin{equation}
\begin{array}{rl}  |C(n)| &= \mathbf{V} \mathbf{D^0} \mathbf{D^1} \cdots \mathbf{V}^T\\
&=\left\{\begin{array}{ll} \mathbf{V} (\mathbf{D}^0 \mathbf{YY} \mathbf{D}^1)^{\frac{n-5}{2}} \mathbf{V}^T & \mbox{ if $n$ is odd};\\
                            \mathbf{V} (\mathbf{D}^0 \mathbf{YY} \mathbf{D}^1)^{\frac{n-6}{2}} \mathbf{D}^0 \mathbf{YY} \mathbf{V}^T & \mbox{ if $n$ is even}; \end{array}\right. \\
&=\mathbf{V} \mathbf{D}^{n-5} \mathbf{YV}^T.
\end{array}
\end{equation}

In the following, we first focus on constraints $(C3,1C)$, $(C4,2C)$, and $(C5,3C)$. The codes based on these constraints are shown to have the same codebooks as OLCs, FPCs, and FOCs, respectively. Then, we consider constraint $(C2,1C)$, which would lead to codes with a smaller delay at the expense of a lower code rate.

\subsection{Codes Under $(C3,1C)$}

The one Lambda codes have a worst-case delay $(1+\lambda)\tau$. According to \cite{Sri07}, the worst-case delay $(1+\lambda)\tau$ can only be achieved \textbf{if and only if} the transitions $\uparrow \downarrow \times$, -$\uparrow$-, and $\uparrow$-$\uparrow$ plus their symmetric and complement versions (e.g. $\uparrow \downarrow\times$ and $\times \downarrow \uparrow$ are symmetric, and -$\downarrow$- is the complement of -$\uparrow$-) are avoided, where $\uparrow$, $\downarrow$, $\times$, and - denote 0$\rightarrow$1, 1$\rightarrow$0, don't care, and no transition, respectively. The first constraint of avoiding $\uparrow \downarrow \times$ ensures that a transition between any two codewords does not cause opposite transition on any wire. This condition is referred as a forbidden-transition (FT) condition. The second constraint of avoiding -$\uparrow$- ensures that 2C patterns are removed. This constraint ensures two adjacent bit boundaries cannot both be 01-type or 10-type, and is referred as a forbidden adjacent boundary pattern (FABP) condition \cite{Sri07}. The last two forbidden patterns give the constraint that no patterns 010 and 101 appear in the codeword, which is referred as a forbidden-pattern (FP) condition \cite{Sri07}. Codes satisfying these \textbf{necessary and sufficient} conditions are called one Lambda codes (OLCs). We denote the largest OLC codebook size for an $n$-bit bus as $G_n$, and $G_n$ is given by
\begin{equation}
G_n=G_{n-1}+G_{n-5}
\label{Eq:4}
\end{equation}
with initial conditions $G_1=2, G_2=3, G_3=4, G_4=5$, and $G_5=7$ \cite{Sri04}.

With our classification, we explore codes under constraint $(C3,1C)$. From Tab.~\ref{tab:cb}, the two largest 5-bit codebooks are given by $C_5^0$=\{0, 3, 14, 15, 24, 30, 31\} and $C_5^1$=\{0, 1, 7, 16, 17, 28, 31\}.
An $n$-bit codebook $C(n)$ can be obtained via Alg.~\ref{alg:generic}.
The number of codewords is given by
\begin{equation}
|C(n)|= \mathbf{V} \mathbf{D}^{n-5}_{(C3,1C)}\mathbf{V}^T \mbox{ for }  n\ge 5,
\label{Eq:C3}
\end{equation}
where $\mathbf{V}$ is a seven-dimensional all one vector and $\mathbf{D}_{(C3,1C)}$ is a $7\times 7$ expansion matrix as shown in Tab.~\ref{tab:expmtx}.
We further establish that the largest codebook sizes under constraint $(C3,1C)$ satisfy the recursion:
\begin{lemma}
  For $n\ge 8$, $|C_{(C3,1C)}(n)|$ is given by a recursion $|C_{(C3,1C)}(n)| = |C_{(C3,1C)}(n-2)| + |C_{(C3,1C)}(n-3)|$, with initial conditions $|C_{(C3,1C)}(n)|=$7, 9, 12, for $n=$5, 6, 7, respectively.
\label{lm:c31c1}
\end{lemma}
See the appendix for the proof. In fact, we can further relate these codes with OLCs by the following:

\begin{thm}
  The codes under ($C3,1C$) have the same codebooks as OLCs. Hence, $G_n=|C_{(C3,1C)}(n)|$.
  \label{thm:c31c}
\end{thm}
See the appendix for the proof. Theorem~\ref{thm:c31c} implies that the codes under constraint $(C3,1C)$ are equivalent to the class of OLC codes.

\subsection{Codes Under $(C4,2C)$}

The $(1+2\lambda)$ codes have a worst-case delay of $(1+2\lambda)\tau$. No necessary and sufficient condition is known for a code to be a $(1+2\lambda)$ code. Two sufficient conditions FT and FP are found, which lead to two families of $(1+2\lambda)$ codes, FTC and FPC, respectively. The size of an FTC codebook for an $n$-wire bus is given by $F_{n+2}$, where $F_n$ is the Fibonacci sequence that satisfies $F_{n+2}=F_{n+1}+F_n$ and has initial conditions $F_1=F_2=1$ \cite{Vic01}. The FPCs for an $n$-wire bus have a larger codebook size $2F_{n+1}$ \cite{Dua01}.

With our classification, we explore codes under constraint $(C4,2C)$. From Tab.~\ref{tab:cb}, only one largest 5-bit codebook is found $C^0_5$=\{0, 1, 3, 6, 7, 12, 14, 15, 16, 17, 19, 24, 25, 28, 30, 31\}.
An $n$-bit codebook $C(n)$ can be obtained via Alg.~\ref{alg:generic} by setting $C^1_5=C^0_5$. The number of codewords is given by
\begin{equation}
  |C(n)|= \mathbf{V}\mathbf{D}^{n-5}_{(C4,2C)} \mathbf{V}^T \mbox{ for } n\ge 5
  \label{Eq:C4}
\end{equation}
where $\mathbf{V}$ is a 16-dimensional all one vector and $\mathbf{D}_{(C4,2C)}$ is a $16\times 16$ expansion matrix as shown in Tab.~\ref{tab:expmtx}.
We further establish that the largest codebook sizes under constraint $(C4,2C)$ satisfy the recursion:
\begin{lemma}
  For $n\ge 9$, $|C_{(C4,2C)}(n)|$ can be simplified as recursion $|C_{(C4,2C)}(n)| = 2|C_{(C4,2C)}(n-1)| -|C_{(C4,2C)}(n-2)| + |C_{(C4,2C)}(n-4)|$, with boundary conditions $|C_{(C4,2C)}(n)|=$16, 26, 42, 68, for $n=$5, 6, 7, 8, respectively.
  \label{lm:c42c1}
\end{lemma}
See the appendix for the proof. Again, we can relate these codes to existing CACs by the following:

\begin{thm}
  The codes under ($C4,2C$) have the same codebooks as FPCs. Hence, $2F_{n+1} = |C_{(C4,2C)}(n)|$.
  \label{thm:c42c}
\end{thm}
See the appendix for the proof. Since FPCs and our codes under $(C4,2C)$ can be obtained by excluding $D3$ plus $D4$ patterns and $C5$ plus $C6$ patterns, respectively, Theorem~\ref{thm:c42c} is not surprising given that $C5$ and $C6$ are the same as $D3$ and $D4$, respectively. Theorem~\ref{thm:c42c} implies that results in the literature regarding FPCs are also applicable to codes under constraint ($C4,2C$).

\subsection{Codes Under $(C5,3C)$}

The $(1+3\lambda)$ codes have a worst-case delay of $(1+3\lambda)\tau$, which can be achieved \textbf{if and only if} $\downarrow \uparrow \downarrow$ and $\uparrow \downarrow \uparrow$ are avoided. So the \textbf{necessary and sufficient} condition for the $(1+3\lambda)$ codes is that the codebook cannot have both 010 and 101 appearing centered around any bit position, which is referred as a forbidden-overlap (FO) condition. Codes satisfying the FO condition are called FOCs. It is shown that the largest FOC codebook for an $n$-bit bus is given by $T_{n+2}$, where $T_n=T_{n-1}+T_{n-2}+T_{n-3}$ is the tribonacci number sequence with initial conditions $T_1=1,$ $T_2=1$, and $T_3=2$ \cite{Sri07}.

With our classification, we explore codes under constraint $(C5,3C)$. Two largest 5-bit codebooks $C_5^0$=\{0, 1, 2, 3, 6, 7, 8, 9, 10, 11, 12, 14, 15, 16, 17, 18, 19, 24, 25, 26, 27, 28, 30, 31\} and $C_5^1$=\{0, 1, 3, 4, 5, 6, 7, 12, 13, 14, 15, 16, 17, 19, 20, 21, 22, 23, 24, 25, 28, 29, 30, 31\} are found. Via Alg.~\ref{alg:generic}, an $n$-bit codebook $C(n)$ can be obtained. The number of codewords is given by
\begin{equation}
  |C(n)|= \mathbf{V}\mathbf{D}^{n-5}_{(C5,3C)} \mathbf{V}^T \mbox{ for }  n\ge 5,
  \label{Eq:C5}
\end{equation}
where $\mathbf{V}$ is a 24-dimensional all one vector and $\mathbf{D}_{(C5,3C)}$ is a $24\times 24$ expansion matrix as shown in Tab.~\ref{tab:expmtx}.

We further establish that the largest codebook sizes under constraint $(C5,3C)$ satisfy the recursion:
\begin{lemma}
  For $n\ge 8$, $|C_{(C5,3C)}(n)|$ can be simplified as recursion $|C_{(C5,3C)}(n)| = |C_{(C5,3C)}(n-1)| -|C_{(C5,3C)}(n-2)| + |C_{(C5,3C)}(n-3)|$, with boundary conditions $|C_{(C5,3C)}(n)|=$24,44,81, for $n=$5, 6, 7, respectively.
  \label{lm:c53c1}
\end{lemma}
See the appendix for the proof. Again we can relate these codes to existing CACs by the following:

\begin{thm}
  The codes under ($C5,3C$) have the same codebooks as FOCs. Hence, $T_{n+2}=|C_{(C5,3C)}(n)|$.
  \label{thm:c53c}
\end{thm}
See the appendix for the proof. Theorem~\ref{thm:c53c} is not surprising, since FOCs and our codes under $(C5,3C)$ can be obtained by excluding $D4$ and $C6$ patterns, respectively, and $D4$ and $C6$ have been shown to be the same. Theorem~\ref{thm:c53c} implies that results in the literature regarding FOCs are also applicable to codes under constraint ($C5,3C$).

\subsection{Codes Under $(C2,1C)$}\label{sec:unpruned}
With our classification, we explore codes under constraint $(C2,1C)$. From Tab.~\ref{tab:cb}, the two largest 5-bit codebooks are given by $C_5^0$=\{00000, 00011, 01111, 11000, 11110, 11111\} and $C_5^1$=\{00000, 00001, 00111, 10000, 11100, 11111\}. An $n$-bit codebook $C(n)$ can be obtained via Alg.~\ref{alg:generic}.
The number of codewords is given by
\begin{equation}
|C(n)|= \mathbf{V} \mathbf{D}^{n-5}\mathbf{V}^T \mbox{ for }  n\ge 5,
\label{Eq:C2}
\end{equation}
where $\mathbf{V}$ is a six-dimensional all one vector and $\mathbf{D} = \left[ \begin{smallmatrix}  0&0&0&0&1&1 \\ 0&0&0&1&0&0 \\ 1&0&0&0&0&0 \\ 0&0&1&0&0&0 \\ 0&1&0&0&0&0 \\ 1&0&0&0&0&0  \end{smallmatrix} \right]$.

We further establish that the largest codebook sizes under constraint $(C2,1C)$ satisfy the recursion:
\begin{lemma}
  For $n\ge 10$, $|C_{(C3,1C)}(n)|$ can be simplified as recursion $|C_{(C2,1C)}(n)| = |C_{(C2,1C)}(n-2)| + |C_{(C2,1C)}(n-5)|$, with initial conditions $|C_{(C2,1C)}(n)|=$6, 7, 9, 11, 14, for $n=$5, 6, 7, 8, 9, respectively.
  \label{lm:c21c1}
\end{lemma}
See the appendix for the proof.

\begin{lemma}
  The codebook under ($C2,1C$) is a subset of OLC.
  \label{lm:c21c2}
\end{lemma}
See the appendix for the proof.

\subsection{Pruned Codes Under $(C2,1C)$}
For $(C2,1C)$, the restriction on the side wires is more relaxed than that on the middle wires, which results in larger worst-case delays for the side wires.
Hence, we prune the CACs under constraint $(C2,1C)$ by removing codewords with larger delays on the side wires in order to achieve a smaller worst-case delay.
Since the pruned codes have  a smaller delay than OLCs, we call these pruned CACs improved one Lambda codes (IOLCs).
We obtain IOLCs by first finding an $n$-bit codebook via Alg.~\ref{alg:generic} as in Sec.~\ref{sec:unpruned}, and then pruning the codebook with Alg.~\ref{alg:IOLC}.
To prune the codebook $C(n)$, we search for maximum subsets of $C^i_5$ ($i=0,1$) with smaller delays on the side wires.
For $C^0_5$, two maximum subsets $C^{0,0}_5$=\{0, 3, 15, 30, 31\} and $C^{0,1}_5$=\{0, 15, 24, 30, 31\}
are found with smaller worst-case delays on wires 1 and 2 and wires 4 and 5, respectively. For $C^1_5$, a maximum subset $C^{1,1}_5$=\{0, 1, 7, 16, 31\}
is found with smaller worst-case delays on wires 4 and 5. Finally, a valid $n$-bit codebook is obtained with the leftmost five bits belonging to $C^{0,0}_5$, and the rightmost five bits belonging to $C^{0,1}_5$ or $C^{1,1}_5$ depending on whether $n$ is odd or even.

\begin{algorithm}[!htb]
  \caption{Pruning CACs under $(C2,1C)$}
  \begin{algorithmic}
    \REQUIRE $C^{0,0}_5$, $C^{0,1}_5$, $C^{1,1}_5$, $C(n)$;
    \IF{$n$ is odd}
    \STATE $i=1$;
    \ELSE
    \STATE $i=0$;
    \ENDIF
    \FOR{$\forall \mathbf{c}_n = (c_1c_2\cdots c_n) \in C(n)$}
    \IF{$(c_1c_2c_3c_4c_5) \not\in C^{0,0}_5$ or $(c_{n-4}c_{n-3}c_{n-2}c_{n-1}c_n) \not\in C^{1-i,1}_5$}
    \STATE eliminate $\mathbf{c}_n$ from $C(n)$;
    \ENDIF
    \ENDFOR
    \RETURN $C(n)$.
  \end{algorithmic}
  \label{alg:IOLC}
\end{algorithm}

The pruning algorithm for CACs under $(C2,1C)$ on an $n$-bit bus is shown in Alg.~\ref{alg:IOLC}. By pruning all codewords $\mathbf{c}_n$ in $C(n)$, the algorithm removes codewords with larger delay on side wires.
With Alg.~\ref{alg:IOLC}, we get an $n$-bit IOLC under constraint $(C2,1C)$, and its size is given by
\begin{equation}
|C_{IOLC}(n)|= \mathbf{W}_1 \mathbf{D}^{n-5} \mathbf{Y} \mathbf{W}_2^T \mbox{ for } n\ge 5,
\label{Eq:C2p}
\end{equation}
where $\mathbf{W}_1=[1 \; 1  \;  1  \;  0  \;  1  \; 1]$, $\mathbf{W}_2=[1  \; 0  \; 1  \; 1  \; 1  \; 1]$, and $\mathbf{D}$ is the same as that in Eq.~(\ref{Eq:C2}). Note that $\mathbf{W}_1$ and $\mathbf{W}_2$ are used instead of $\mathbf{V}$, because of the pruning of valid patterns on side wires.

We further establish that the largest codebook sizes of IOLCs satisfy the recursion:
\begin{lemma}
  For $n\ge 10$, $|C_{IOLC}(n)|$ can be simplified as recursion $|C_{IOLC}(n)| = |C_{IOLC}(n-2)| + |C_{IOLC}(n-5)|$, with initial conditions $|C_{IOLC}(n)|=$4, 5, 7, 8, 11, for $n=$5, 6, 7, 8, 9, respectively.
  \label{lm:IOLC1}
\end{lemma}
This recursion is the same as that in that in Lemma~\ref{lm:c21c1}. It can be proved in the same fashion as for Lemma~\ref{lm:c21c1}, and hence its proof is omitted.

\begin{lemma}
  The IOLC codebook is a subset of OLC.
  \label{lm:IOLC2}
\end{lemma}
See the appendix for the proof.

\begin{table}[!h]
\caption{Simulated delays of our IOLC, unpruned ($C2,1C$) code, and OLC~\cite{Dua04} for a 10-bit bus ($\lambda=12.24$ and $\tau_0=1.42$ps).}\label{tab:10bit}\vspace{-0.1in}
\begin{center}
\begin{tabular}{|c|c|c|c|}
\hline
\multirow{2}{*}{Wire $i$} & \multicolumn{3}{|c|}{Delays (ps)}\\
\cline{2-4} & IOLCs & ($C2,1C$) & OLCs\\
\hline
1 & 10.08 & 5.49 & 10.55 \\
\hline
2 & 7.03 & 9.13 & 2.92 \\
\hline
3 & 9.31 & 9.31 & 5.94 \\
\hline
4 & 9.31 & 9.45 & 6.09 \\
\hline
5 & 9.59 & 9.36 & 10.73 \\
\hline
6 & 9.41 & 9.41 & 13.64\\
\hline
7 & \textbf{10.14} & 10.14 & 14.06\\
\hline
8 & 9.65 & 10.57 & 14.84\\
\hline
9 & 8.97 & 9.14 & 8.99\\
\hline
10 & 5.28 & \textbf{13.50} & \textbf{14.84}\\
\hline
\end{tabular}
\end{center}
\end{table}

\begin{table}[!h]
\caption{Simulated delays of our IOLC, unpruned ($C2,1C$) code, and OLC~\cite{Dua04} for a 16-bit bus ($\lambda=12.24$ and $\tau_0=1.42$ps).}\label{tab:16bit}\vspace{-0.1in}
\begin{center}
\begin{tabular}{|c|c|c|c|}
\hline
\multirow{2}{*}{Wire $i$} & \multicolumn{3}{|c|}{Delays (ps)}\\
\cline{2-4} & IOLCs & ($C2,1C$) & OLCs\\
\hline
1 & 10.32 & \textbf{13.92} & 15.95 \\
\hline
2 & 7.43 & 9.51 & 10.03 \\
\hline
3 & 9.57 & 10.88 & 15.54 \\
\hline
4 & 9.83 & 10.21 & 15.75 \\
\hline
5 & 10.16 & 10.16 & 15.02 \\
\hline
6 & 10.33 & 10.34 & 15.57 \\
\hline
7 & 10.39 & 10.39 & 15.70 \\
\hline
8 & 10.23 & 10.23 & 15.48 \\
\hline
9 & 9.87 & 10.25 & 15.57\\
\hline
10 & \textbf{10.40} & 10.39 & 15.66\\
\hline
11 & 10.34 & 10.33 & 15.52\\
\hline
12 & 10.17 & 10.21 & 14.88\\
\hline
13 & 10.25 & 10.39 & 15.85\\
\hline
14 & 9.98 & 10.92 & 15.59\\
\hline
15 & 9.61 & 9.62 & 10.13\\
\hline
16 & 5.58 & 13.92 & \textbf{16.11}\\
\hline
\end{tabular}
\end{center}
\end{table}

\vspace{-0.1in}
\section{Performance Evaluation}
\label{sec:performance}
In this section, we evaluate the performance of CACs based on our classification with extensive simulations, and compare them with existing CACs.
Each CAC has two key performance metrics: delay and rate.
The delay of a CAC is the worst-case delay when the codewords from the CAC are transmitted over the bus. Codebook size and code rate are often used to measure the overhead of CACs. The codebook size of a CAC is simply the number of codewords. Suppose a CAC of size $M$ is transmitted over an $n$-bit bus, then its rate is defined as $\frac{\left\lfloor \log_2 M \right\rfloor}{n}$. A CAC of rate $k/n$ implies that $n-k$ extra wires are used in addition to $k$ data wires so as to reduce the crosstalk delay. Hence, the code rate measures the area and power overhead of CACs: the higher the rate, the smaller the overhead.
Obviously, there is a tradeoff between the code rate and delay of a CAC: typically a lower rate code is needed to achieve a smaller delay.
To measure the overall effects of both rate and delay, we also define the throughput of a CAC as the ratio of code rate and delay. The assumptions for this definition are: (1) the clock rate of the bus is determined by the inverse of the worst-case delay; (2) the throughput of the bus is linearly proportional to $k$, the number of data wires.

Since codes under ($C3,1C$), ($C4,2C$), and ($C5,3C$) have exactly the same codebooks as OLCs, FPCs, and FOCs, their delay, rate, and throughput are also the same.
Under constraint ($C2,1C$), we propose two kinds of codes, unpruned codes and pruned codes (IOLCs). In the following, we compare their performance with OLCs in \cite{Dua04} with extensive simulations.

\begin{table*}[!th]
\caption{Comparison of codebook size and throughput of IOLC, unpruned ($C2,1C$) code, and OLC~\cite{Dua04} ($\lambda=12.24$ and $\tau_0=1.42$ps).}\label{tab:2}\vspace{-0.1in}
\begin{center}
\begin{tabular}{|c|c|c|c|c|c|c|c|c|}
\hline
\# of & \multicolumn{3}{|c|}{IOLC} & \multicolumn{3}{|c|}{($C2,1C$)} & \multicolumn{2}{|c|}{OLC}\\
\cline{2-9} wires & \# of words & \# of bits & Throughput Gain & \# of words & \# of bits & Throughput Gain & \# of words & \# of bits \\
\hline
5  & 4 & 2 & \textbf{1.55} & 6 & 2 & \textbf{1.10} & 7 & 2\\
\hline
6  & 5 & 2 & 1.07 & 7 & 2 & 0.78 & 9 & 3\\
\hline
7  & 7 & 2 & 1.02 & 9 & 3 & 1.14 & 12 & 3\\
\hline
8  & 8 & 3 & 1.12 & 11 & 3 & 0.84 & 16 & 4\\
\hline
9  & 11 & 3 & 1.10 & 14 & 3 & 0.84 & 21 & 4\\
\hline
10  & 12 & 3 & 1.10 & 17 & 4 & \textbf{1.10} & 28 & 4\\
\hline
11  & 16 & 4 & 1.18 & 21 & 4 & 0.88 & 37 & 5\\
\hline
12  & 18 & 4 & 1.19 & 26 & 4 & 0.89 & 49 & 5\\
\hline
13  & 23 & 4 & 1.03 & 32 & 5 & 0.96 & 65 & 6\\
\hline
14  & 27 & 4 & 1.02 & 40 & 5 & 0.95 & 86 & 6\\
\hline
15  & 34 & 5 & 1.27 & 49 & 5 & 0.95 & 114 & 6\\
\hline
16  & 41 & 5 & 1.11 & 61 & 5 & 0.83 & 151 & 7\\
\hline
\end{tabular}
\end{center}\vspace{-0.1in}
\end{table*}

To compare the worst-case delay of our IOLCs, unpruned ($C2,1C$) codes, and OLCs, we simulate two buses, a 10-bit bus and a 16-bit bus, with all transitions between any two codewords in their codebooks and obtain the worst-case delays of each wire.
The simulation environment has been explained in Sec.~\ref{sec:pat}. Both buses have a length of 5mm, and $\tau_0=1.42$ps and $\lambda=12.24$.
The simulation results are shown in Tabs.~\ref{tab:10bit} and \ref{tab:16bit}, where for each CAC the largest delays among all wires are in boldface. As commented above for unpruned $(C2,1C)$ codes, the delays of the two outmost wires are significantly greater than those of other wires.
For a 10-bit bus, the worst-case delays of our IOLC, unpruned ($C2,1C$) code, and an OLC are given by 10.14ps, 13.50ps, and 14.84ps, respectively. The worst-case delay of our IOLC and unpruned ($C2,1C$) code are 31.67\% and 9.03\% smaller than that of the OLC, respectively.
For a 16-bit bus, the worst-case delays of our IOLC, unpruned ($C2,1C$) code, and an OLC are given by 10.40ps, 13.92ps, and 16.11ps, respectively. The worst-case delay of our IOLC and unpruned ($C2,1C$) code are 35.44\% and 13.59\% smaller than that of the OLC, respectively.

For all simulations, our IOLCs have better delay performance than OLCs. Although both IOLCs and unpruned ($C2,1C$) codes have almost the same code rate and better delay performance than OLCs, the delay performance of IOLCs is much better than the unpruned ($C2,1C$) codes. With a more advanced technology where the coupling effect is significant, the improvement of our IOLCs is bigger.

The comparisons of the codebook size between our IOLCs, unpruned ($C2,1C$) codes, and OLCs~\cite{Dua04} and the throughput gain with respect to OLCs are shown in Tab.~\ref{tab:2}. The throughput gain of our CACs with respect to OLCs is given by the ratio between the throughput of our CACs and the throughput of OLCs.
The codebook sizes of the three codes are close. In all cases, the difference of the number of bits between our IOLCs and unpruned ($C2,1C$) codes is within 1 bit. The difference of the number of bits between our IOLCs and OLCs~\cite{Dua04} is within 2 bits for $n\le 16$.
In respect to throughput, our IOLCs always have a greater throughput than OLCs, and their throughput gain ranges from 1.02 to 1.55 for an $n$-wire bus ($5\le n \le 16$). The unpruned $(C2,1C$) codes have better throughput in some cases than OLCs, and the throughput gain ranges from 0.78 to 1.10 for an $n$-wire bus ($5\le n \le 16$).
When unpruned $(C2,1C)$ codes have a lower throughput than OLCs, IOLCs can be used.

Our IOLCs and unpruned ($C2,1C$) codes provide additional options for the tradeoff between code rate and code delay.
In addition to achieving higher throughputs, the new CACs are also appropriate for interconnects where the delay is of top priority.

It has been shown that the encoding and decoding of OLCs, FPCs, and FOCs have quadratic complexity based on numeral systems \cite{WY_TVLSI11}. Since codes under ($C3,1C$), ($C4,2C$), and ($C5,3C$) have exactly the same codebooks as OLCs, FPCs, and FOCs, their CODECs also have quadratic complexity. Also, it is expected that the encoding and decoding of our IOLCs and unpruned ($C2,1C$) codes have a quadratic complexity, since the codebooks of our IOLCs and unpruned ($C2,1C$) codes are proper subsets of OLCs.

We remark that the simulation results in Sections~\ref{sec:pat} and \ref{sec:performance} are all based on a 45nm CMOS technology. We have also run the same set of simulations based on a 0.1-$\mu$m technology (omitted for brevity).  Between the two sets of simulation results, the main conclusions of the manuscript and the key features of our proposed classification and CACs remain the same. For instance, the delays of the patterns in different classes do not overlap, regardless of the technology. Also, the proposed CACs based on the new classification are also the same. This actually demonstrates that our approach to delay classification and CACs is applicable to a wide variety of technology. This is because in our approach, the dependency of the crosstalk delay on the technology is represented by the two parameters, the propagation delay $\tau_0$ of a wire free of crosstalk and the coupling factor $\lambda$. Since our analytical approach to the classification and CACs treats these two parameters as variables, our approach can be easily adapted to a wide variety of technology.

\section{CONCLUSIONS}\label{sec:conclusion}
In this paper, we propose a new classification of transition patterns. The new classification has finer classes and the delays do not overlap among different classes. Hence the new classification is conducive to the design of CACs. To illustrate this, we design a family of CACs with different constraints. Some codes of the family are the same as existing codes, OLCs, FPCs, and FOCs. We also propose two new CACs with a smaller worst-case delay and better throughput than OLCs. Since our analytical approach to the classification and CACs treats  the technology-dependent parameters as variables, our approach can be easily adapted to a wide variety of technology.

\appendix

\begin{proof}[Proof of Lemma~\ref{lm:c31c1}]
The eigenvalues of $\mathbf{D}$ are given by solving $\det |\lambda \mathbf{I} - \mathbf{D}|=0$. Then,
\begin{equation*}
  \begin{array}{rl} & \det |\lambda \mathbf{I} - \mathbf{D}|=0\\
  \Rightarrow & \lambda^7 - \lambda^5 - \lambda^4 = 0\\
  \Rightarrow & \mathbf{D}^7 = \mathbf{D}^5 - \mathbf{D}^4\\
  \Rightarrow & \mathbf{V D}^7 \mathbf{V}^T= \mathbf{V D}^5 \mathbf{V}^T + \mathbf{V D}^4 \mathbf{V}^T\\
  \Rightarrow & |C(n)|=|C(n-2)|+|C(n-3)|.
  \end{array}
\end{equation*}

For $n=5,6,7$, the boundary conditions can be obtained by Eq.~(\ref{Eq:C3}) as $|C(5)|=7$, $|C(6)|=9$, and $|C(7)|=12$. Thus, the lemma holds for $n \ge 8$.
\end{proof}

\begin{proof}[Proof of Theorem~\ref{thm:c31c}]
It has been shown that an $(n+1)$-bit OLC codebook $C(n+1)$ can be constructed from an $n$-bit codebook $C(n)$ \cite{Dua04}. The necessary and sufficient condition for OLCs defines the same expansion matrix as our codes. The OLC construction is the same as that of our codes under $(C3,1C)$ shown in Alg.~\ref{alg:generic}. For $n=5$, the OLC codebooks are the same as our codes under $(C3,1C)$. So, for an $n$-bit bus ($n \ge 5$), codes under constraint $(C3,1C)$ are the same as OLCs. For an $n$-bit bus ($n\le 4$), the constraint $(C3,1C)$ reduces to $1C$, and leads to the same codebooks as OLCs. Hence, our codes under ($C3,1C$) have the same codebooks as OLCs, which implies that $G_n = |C(n)|$.

\end{proof}

\begin{proof}[Proof of Lemma~\ref{lm:c42c1}]
The eigenvalues of $\mathbf{D}$ are given by solving $\det |\lambda \mathbf{I} - \mathbf{D}|=0$. Then,
\begin{equation*}
  \begin{array}{rl} & \det |\lambda \mathbf{I} - \mathbf{D}|=0\\
  \Rightarrow & \mathbf{D}^{16} = 2\mathbf{D}^{15} - \mathbf{D}^{14} + \mathbf{D}^{12}\\
  \Rightarrow & \mathbf{V D}^{16} \mathbf{V}^T= 2 \mathbf{V D}^{15} \mathbf{V}^T - \mathbf{V D}^{14} \mathbf{V}^T + \mathbf{V D}^{12} \mathbf{V}^T\\
  \Rightarrow & |C(n)|=2 |C(n-1)| - |C(n-2)| + |C(n-4)|.
  \end{array}
\end{equation*}

For $n=5,6,7,8$, the boundary conditions can be obtained by Eq.~(\ref{Eq:C4}) as $|C(5)|=16$, $|C(6)|=26$, $|C(7)|=42$, and $|C(8)|=68$. Thus, the lemma holds for $n \ge 9$.
\end{proof}

\begin{proof}[Proof of Theorem~\ref{thm:c42c}]
It has been shown that an $(n+1)$-bit FPC codebook $C(n+1)$ can be constructed from an $n$-bit codebook $C(n)$ \cite{Dua01}. The sufficient condition (FP condition) for FPCs defines the same expansion matrix as our codes. The FPC construction is the same as that of our codes under $(C4,2C)$ shown in Alg.~\ref{alg:generic}. For $n=5$, the FPC codebooks are the same as our codes under $(C4,2C)$. So, for an $n$-bit bus ($n \ge 5$), codes under constraint $(C4,2C)$ are the same as FPCs. For an $n$-bit bus ($n\le 4$), the constraint $(C4,2C)$ reduces to $2C$, and leads to the same codebooks as FPCs. Hence, our codes under ($C4,2C$) have the same codebooks as FPCs, which implies that $2F_{n+1} = |C(n)|$.

\end{proof}

\begin{proof}[Proof of Lemma~\ref{lm:c53c1}]
The eigenvalues of $\mathbf{D}$ are given by solving $\det |\lambda \mathbf{I} - \mathbf{D}|=0$. Then,
\begin{equation*}
  \begin{array}{rl} & \det |\lambda \mathbf{I} - \mathbf{D}|=0\\
  \Rightarrow & \mathbf{D}^{24} = \mathbf{D}^{23} + \mathbf{D}^{22} + \mathbf{D}^{21}\\
  \Rightarrow & \mathbf{V D}^{24} \mathbf{V}^T= \mathbf{V D}^{23} \mathbf{V}^T + \mathbf{V D}^{22} \mathbf{V}^T + \mathbf{V D}^{21} \mathbf{V}^T\\
  \Rightarrow & |C(n)|=|C(n-1)| + |C(n-2)| + |C(n-3)|.
  \end{array}
\end{equation*}

For $n=5,6,7,8$, the boundary conditions can be obtained by Eq.~(\ref{Eq:C5}) as $|C(5)|=24$, $|C(6)|=44$, and $|C(7)|=81$. Thus, the lemma holds for $n \ge 9$.
\end{proof}

\begin{proof}[Proof of Theorem~\ref{thm:c53c}]
It has been shown that an $(n+1)$-bit FOC codebook $C(n+1)$ can be constructed from an $n$-bit codebook $C(n)$ \cite{Dua01}. The necessary and sufficient condition (FO condition) for FOCs defines the same expansion matrix as our codes. The FOC construction is the same as that of our codes under $(C5,3C)$ shown in Alg.~\ref{alg:generic}. For $n=5$, the FOC codebooks are the same as our codes under $(C5,3C)$. So, for an $n$-bit bus ($n \ge 5$), codes under constraint $(C5,3C)$ are the same as FOCs. For an $n$-bit bus ($n\le 4$), the constraint $(C5,3C)$ reduces to $3C$, and leads to the same codebooks as FOCs. Hence, our codes under ($C5,3C$) have the same codebooks as FOCs, which implies that $T_{n+2} = |C(n)|$.
\end{proof}

\begin{proof}[Proof of Lemma~\ref{lm:c21c1}]
The eigenvalues of $\mathbf{D}$ are given by solving $\det |\lambda \mathbf{I} - \mathbf{D}|=0$. Then,
\begin{equation*}
  \begin{array}{rl} & \det |\lambda \mathbf{I} - \mathbf{D}|=0\\
  \Rightarrow & \mathbf{D}^6 = \mathbf{D}^4 - \mathbf{D}\\
  \Rightarrow & \mathbf{V D}^6 \mathbf{V}^T= \mathbf{V} \mathbf{D}^4 \mathbf{V}^T + \mathbf{V D} \mathbf{V}^T\\
  \Rightarrow & |C(n)|=|C(n-2)|+|C(n-5)|.
  \end{array}
\end{equation*}

For $n=5,6,7,8,9$, the boundary conditions can be obtained by Eq.~(\ref{Eq:C2}) as $|C(5)|=6$, $|C(6)|=7$, $|C(7)|=9$, $|C(8)|=11$, and $|C(9)|=14$. Thus, the lemma holds for $n \ge 10$.
\end{proof}

\begin{proof}[Proof of Lemma~\ref{lm:c21c2}]
As shown in Tab.~\ref{tab:cb}, $C_5^i$ under ($C2,1C$) is a subset of $C_5^i$ under ($C3,1C$) for $i=0,1$. Thus, the valid expansions from $C_5^i$ to $C_5^{1-i}$ under ($C2,1C$) is part of that under ($C3,1C$). So, for an $n$-bit bus, $C_{(C2,1C)}(n)\subset C_{(C3,1C)}(n)$. According to Thm.~\ref{thm:c31c}, the $n$-bit codebook $C_{(C2,1C)}(n)$ is a subset of an OLC codebook.
\end{proof}

\begin{proof}[Proof of Lemma~\ref{lm:IOLC2}]
Since the IOLC codebook is a subset of the unpruned codes under ($C2,1C$), this follows Lemma~\ref{lm:c21c2}.
\end{proof}

\balance

\bibliographystyle{IEEEtran}
\bibliography{bibbus}

\end{document}